\documentclass[10pt,notitlepage,a4paper]{article}
\usepackage{authblk}

\usepackage{dune}
\usepackage{pgfplots}

\usepackage[english]{babel}
\usepackage{amsthm,amssymb,amsmath}
\usepackage{enumerate}
\usepackage{setspace}
\usepackage[verbose,section]{placeins}

\usepackage{tikz}
\usetikzlibrary{patterns}
\usetikzlibrary{decorations.pathreplacing}
\usetikzlibrary{arrows,positioning} 
\usepackage{tikz-qtree}

\usepackage{comment}
\usepackage{color}
\usepackage{xspace}
\usepackage{graphicx,epstopdf}
\usepackage{array}
\usepackage{multirow}

\usepackage[ruled,algosection,linesnumbered]{algorithm2e}
\usepackage{algorithmic}

\usepackage[bookmarks=true]{hyperref}

\numberwithin{figure}{section}
\numberwithin{equation}{section}

\renewenvironment{itemize}
{\begin{list}{$\bullet$}{\labelwidth0mm \leftmargin3mm %
  \itemsep0pt plus 0pt \topsep3pt \parsep1pt plus 4pt \labelsep2mm}}
{\end{list}}

\newtheorem{theorem}{Theorem}[section]

\newtheorem{proposition}[theorem]{Proposition}

\newtheorem{lemma}[theorem]{Lemma}

\theoremstyle{definition}
\newtheorem{definition}[theorem]{Definition}
\newtheorem{assumption}[theorem]{Assumption}

\theoremstyle{remark}
\newtheorem{remark}[theorem]{Remark}
\newtheorem{example}[theorem]{Example}

\numberwithin{figure}{section}
\numberwithin{equation}{section}

\makeatletter
\newcommand*\rel@kern[1]{\kern#1\dimexpr\macc@kerna}
\newcommand*\widebar[1]{%
  \begingroup
  \def\mathaccent##1##2{%
    \rel@kern{0.8}%
    \overline{\rel@kern{-0.8}\macc@nucleus\rel@kern{0.2}}%
    \rel@kern{-0.2}%
  }%
  \macc@depth\@ne
  \let\math@bgroup\@empty \let\math@egroup\macc@set@skewchar
  \mathsurround\z@ \frozen@everymath{\mathgroup\macc@group\relax}%
  \macc@set@skewchar\relax
  \let\mathaccentV\macc@nested@a
  \macc@nested@a\relax111{#1}%
  \endgroup
}
\makeatother

\newcommand{\abs}[1]{\ensuremath{\left|#1\right|}}

\newcommand{\Ceq}{\ensuremath{\eqsim}}

\newcommand{\definedas}{\mathrel{:=}}

\DeclareMathOperator{\diam}{diam}

\newcommand{\edge}{\ensuremath{E}}
\newcommand{\edges}{\ensuremath{\mathcal{E}}}
\newcommand{\edgesk}[1][k]{\ensuremath{\edges_{#1}}}

\newcommand{\elm}{\ensuremath{T}\xspace}

\let\Forall=\forall
\renewcommand{\forall}{\Forall\,}
\newcommand{\forest}{\ensuremath{\mathcal{F}}}

\DeclareMathOperator{\gen}{gen}

\newcommand{\grids}{\ensuremath{\mathbb{T}}\xspace}
\newcommand{\grid}{\mathcal{T}}
\newcommand{\gridk}[1][k]{\grid_{#1}}

\newcommand{\gridz}[1][z]{\grid(#1)}

\newcommand{\hG}[1][\grid]{\ensuremath{h_{#1}}}
\newcommand{\hE}[1][\elm]{\ensuremath{h_{#1}}}

\newcommand{\ie}{\hbox{i.\,e.},\xspace}

\newcommand{\N}{\ensuremath{\mathbb{N}}}

\newcommand{\normM}[2][]{%
  #1\vert\kern-0.9pt#1\vert\kern-0.9pt#1\vert #2
  #1\vert\kern-0.9pt#1\vert\kern-0.9pt#1\vert}
\newcommand{\NVB}{\textsf{NVB}\xspace}

\newcommand{\Omegaz}[1][z]{\Omega_{#1}}

\let\Paragraph=\paragraph
\renewcommand{\paragraph}[1]{\Paragraph{\textbf{#1.}}}

\newcommand{\R}{\ensuremath{\mathbb{R}}}

\newcommand{\vertices}{\ensuremath{\mathcal{V}}}

\newcommand{\verticesk}[1][k]{\ensuremath{\mathcal{V}_{#1}}}
\newcommand{\vol}[1]{\abs{#1}}

\begin{document}

\newcommand{\email}[1]{\texttt{#1}}

\title{Distributed Newest Vertex Bisection}

\author[a]{Martin Alk\"amper}

\author[b]{Robert Kl\"ofkorn}

\affil[a]{Institut f\"ur Angewandte Analysis und
  Numerische Simulation, Fachbereich Mathematik,
  Universit\"at Stuttgart,
  Pfaffenwaldring 57, D-70569 Stutt\-gart, Germany,
  \url{www.ians.uni-stuttgart.de/nmh/}, \email{alkaemper@ians.uni-stuttgart.de}}%

\affil[b]{International Research Institute of Stavanger, 
P. O. Box 8046, 4068 Stavanger, Norway,
\url{http://www.iris.no}, \email{robert.kloefkorn@iris.no}}

\maketitle

\providecommand{\keywords}[1]{\textbf{Keywords: } #1}



\begin{abstract}
Distributed adaptive conforming refinement requires multiple iterations of the serial refinement algorithm and global communication as the refinement can be propagated over several processor boundaries. We show bounds on the maximum number of iterations. The algorithm is implemented within the software package \dune[ALUGrid].  \\

\noindent
\end{abstract}

\section{Introduction\label{S:intro}}

Conforming finite elements over conforming, unstructured, adaptive grids have been shown to behave very well for numerical simulations of diffusive processes (\cite{BonitoNochetto:10}).  
On the other hand new computer architectures demand parallelism in algorithms and grids to reach their full potential. For an adaptive, parallel, unstructured and conforming grid we need a parallel (or distributed) refinement strategy. \\
In this paper we analyze the distributed refinement strategy called Distributed Newest Vertex Bisection. 
It is the straightforward extension \cite{Zhang:08} of the serial Newest Vertex 
Bisection (\NVB) introduced by Sewell \cite{sewell:72} and we will 
show that the parallel overhead is bounded, in particular by a 
constant independent of the number of processors (Theorem \ref{theorem:constant}). 
The Distributed \NVB as described in this paper 
has been implemented in the 
open-source package \dune[ALUGrid]~\cite{alugrid:16}  
which is a module for the \dune software framework~\cite{dunepaperII:08,dunepaperI:08}.
Domain decomposition in combination with adaptivity requires load balancing 
to equidistribute the workload. 
In \dune[ALUGrid] this is done by equilibrating the number of cells belonging to each processor. 
We make the common assumption that load balancing is done after the 
refinement algorithm is finished. 
Hence we will not consider its effect on the computational cost of the refinement algorithm.  
However, we will show that it is possible to implement the Distributed \NVB 
using modern techniques of parallel computing such as communication hiding 
to achieve excellent strong scaling on a petascale super computer. \\  

Most implementations of parallel adaptive grids use nonconforming hexahedral (or quadrilateral in 2 dimensions) cells with some mesh-balance to acquire a mesh grading required for stability estimates. Implementations of conforming adaptive parallel meshes are scarce, as especially in more than two dimensions the refinement propagation is non-trivial.\\
The Distributed \NVB refinement strategy is also implemented in the toolbox AMDiS \cite{amdis:15}, and we will show, that the bound they give on the communication ($O(\log P)$, where $P$ is the number of partitions in \cite[Sec. 2.4]{amdis:15}) is too weak for large $P$, 
unless the decomposition fulfills additional assumptions.\\
Another approach to parallel simplex refinement has been published by Rivara et al. \cite{rivara:06}. 
In this work a parallel algorithm is introduced that produces an unstructured conforming mesh, 
which is not parallel in the domain decomposition sense, but instead the whole grid 
is known on each processor and the algorithm itself is executed in parallel. 
Furthermore, the refinement strategy differs, as it is not Newest Vertex Bisection(\NVB), but Longest Edge Bisection which was introduced by Rivara in \cite{rivara:84}. 
In \cite{dolfin} a fully distributed parallel Longest Edge Bisection is implemented based on the package 
DOLFIN and reasonable scaling results up to $1024$ cores are presented, however, by
sacrificing the theoretical backing of the adaptation algorithm.

The rest of the paper is structured as follows. First we introduce \NVB with examples 
for $2$-dimensional grids. Then the \NVB refinement algorithm 
is extended to work in decomposed domains. Afterwards we analyze the 
Distributed \NVB and show bounds on the parallel overhead which 
are reflected in the numerical experiments. 
The results hold for grids of any dimension unless stated otherwise.

\medskip

\section{Newest Vertex Bisection\label{S:NVB}}

In this section we shortly introduce \NVB for
conforming triangulations in two space dimensions. It was 
introduced by Sewell \cite{sewell:72} and enhanced by Mitchell with a recursive refinement
algorithm \cite{Mitchell:88,Mitchell:89}; compare also with
\cite{Baensch:91,Maubach:95,Stevenson:08,Traxler:97}. We follow the notation of \cite{GaHeSi:15}.


\subsection{Recurrent bisection of a simplex}

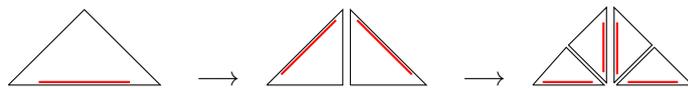
\begin{figure}[hbtp]
  \centering
  \begin{tikzpicture}[scale=2]
    \draw (0,0)--(0.5,0.5) -- (1,0)-- (0,0);
    \draw[red,thick] (0.2,0.02)--(0.8,0.02) ;
  \end{tikzpicture}
  \quad$\longrightarrow$\quad
  \begin{tikzpicture}[scale=2]
    \draw (-0.025,0)--(0.475,0.5) -- (0.475,0)-- (-0.025,0);
    \draw[red,thick] (0.07,0.07)--(0.43,0.43) ;
    \draw (0.525,0)--(0.525,0.5) -- (1.025,0)-- (0.525,0);
    \draw[red,thick] (0.57,0.43) -- (0.93,0.07) ;
  \end{tikzpicture}
  \quad$\longrightarrow$\quad
  \begin{tikzpicture}[scale=2]
    \draw (-0.01,0)--(0.21,0.25) -- (0.46,0)-- (-0.01,0);
    \draw[red,thick] (0.055,0.02)--(0.385,0.02) ;
    \draw (0.225,0.265)--(0.475,0.515) -- (0.475,0.015)-- (0.225,0.265);
    \draw[red,thick] (0.455,0.415) -- (0.455,0.085) ;
    \draw (0.525,0.015)--(0.525,0.515) -- (0.775,0.265)-- (0.525,0.015);
    \draw[red,thick] (0.545,0.415) -- (0.545,0.0715) ;
    \draw (0.54,0)--(0.79,0.25) -- (1.04,0)-- (0.54,0);
    \draw[red,thick] (0.615,0.02)--(0.945,0.02) ;
  \end{tikzpicture}
  \vskip-2mm

  \caption{\NVB: a triangle with its two children and four grandchildren.
    The refinement edges are indicated in red.}\label{F:NVB}
\end{figure}

In order to easily describe \NVB we
identify a simplex $\elm$ with its set of \emph{ordered} vertices
\begin{displaymath}
  \elm = [z_0,z_1,z_2].
\end{displaymath}
The edge between the first and last vertex we call \emph{refinement
  edge}. \NVB refines $\elm$ by inserting a new
vertex in the midpoint $\bar z = \frac12(z_0+z_2)$ of the refinement
edge $\overline{z_0z_2}$ and 
\begin{displaymath}
  \elm_1 = [z_0,\bar{z},z_1]\qquad \text{and}\qquad \elm_2 = [z_2,\bar{z},z_1]
\end{displaymath}
are the two children of $\elm$.  This procedure automatically presets
the children's refinement edges by the local ordering of their
vertices. \NVB thereby determines the refinement edge of any
descendant produced by recurrent bisection of a given initial element
$\elm_0$ from the vertex order of $\elm_0$; see Figure~\ref{F:NVB}.

Recurrent bisection induces the structure of an \emph{infinite binary
  tree} $\forest(\elm_0)$: Any node $\elm$ inside the tree is an
element generated by recurrent application of \NVB.  The two
successors of a node $\elm$ are the children $\elm_1,\elm_2$ created
by a applying \NVB to $\elm$.

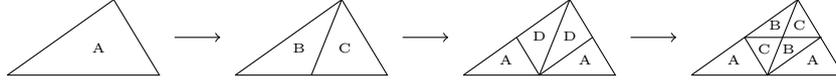
\begin{figure}[hbtp]
  \centering
  \begin{tikzpicture}[scale=2,font=\tiny]
    \draw (0,0)--(0.7,0.5) -- (1,0)-- (0,0);
    \draw[->] (1.1,0.25) -- (1.4,0.25);
    \node at (0.6,0.18) {A};
    \draw (1.5,0)--(2.2,0.5) -- (2.5,0)-- (1.5,0);
    \draw (2,0,0)--(2.2,0.5);
    \draw[->] (2.6,0.25) -- (2.9,0.25);
    \node at (1.92,0.18) {B};
    \node at (2.22,0.18) {C};
    \draw (3,0)--(3.7,0.5) -- (4,0)-- (3,0);
    \draw (3.5,0.0)--(3.7,0.5);
    \draw (3.35,0.25) -- (3.5,0.0) -- (3.85,0.25);
    \draw[->] (4.1,0.25) -- (4.4,0.25);
    \node at (3.28,0.1) {A};
    \node at (3.80,0.1) {A};
    \node at (3.5,0.26) {D};
    \node at (3.7,0.26) {D};
    \draw (4.5,0)--(5.2,0.5) -- (5.5,0)-- (4.5,0);
    \draw (5.0,0.0)--(5.2,0.5); 
    \draw (4.85,0.25) -- (5.0,0.0) -- (5.35,0.25) -- (4.85,0.25);
    \node at (4.78,0.1) {A};
    \node at (5.30,0.1) {A};
    \node at (4.98,0.17) {C};
    \node at (5.14,0.17) {B};
    \node at (5.05,0.33) {B};
    \node at (5.21,0.33) {C};
  \end{tikzpicture}
  \vskip-2mm

  \caption{\NVB: The four similarity classes for an initial element $\elm_0$.}
  \label{F:NVB-shape}
\end{figure}

Finally, \NVB produces shape regular descendants
since all $\elm\in\forest(\elm_0)$ belong to at most four similarity
classes; compare with Figure~\ref{F:NVB-shape}. This is a consequence
of the fact that \NVB always bisects the angle at the newest
vertex. In the end, any angle of any simplex is bisected at most
once. This can be extended to any dimension $d$, see e.g. \cite{Stevenson:08}.
In the following, unless specified explicitly, we assume a general $d \geq 2$.
\subsection{Recurrent refinement of triangulations with \NVB}

Let $\gridk[0]$ be a conforming and exact triangulation of a bounded
polygon $\Omega\subset\R^d$. We can refine $\gridk[0]$ or a refinement
$\grid$ of $\gridk[0]$ by applying the \NVB to selected simplices. 
More than the selected elements have to be refined when striking for
conforming triangulations. Here we refer to \cite{Mitchell:88} for a
recursive refinement algorithm and to \cite{Baensch:91} for an
iterative one.

We next introduce notations related to triangulations. The \emph{master forest} 
\begin{displaymath}
  \forest \definedas \forest(\gridk[0]) = \bigcup_{\elm_0 \in \gridk[0]} \forest(\elm_0)
\end{displaymath}
holds full information about all possible refinements of
$\gridk[0]$. We denote by $\grids=\grids(\gridk[0])$ the class of all
\emph{conforming} refinements of $\gridk[0]$.

For $\grid\in\grids$ the sets of all its vertices and edges are
$\vertices$ and $\edges$, respectively. For $\elm\in\grid$ we set
$\vertices(\elm)\definedas\vertices\cap\elm$ and for $z\in\vertices$ we define
$\gridz\definedas\{\elm\in\grid\mid z\in\elm\}$. The finite element star at a
vertex $z$ is then $\Omegaz\definedas\bigcup\{\elm \colon \elm\in\gridz\}$.  We
let $\hG\in L_\infty(\Omega)$ be the piecewise constant mesh-size
function with
${\hG}_{|\elm}=\hE\definedas\vol{\elm}^{1/2}\Ceq\diam(\elm)$ for
$\elm\in\grid$.  We use $h_{\min,\max}(\grid)$ for the smallest and
largest element size of $\grid$. We say $\elm,\elm'\in\grid$ are
\emph{direct neighbours} iff there is an $\edge \in \edges$ with $\edge \subset \elm\cap\elm'$.

Important in the course of this article is the \emph{generation} of an
element. For each $\elm\in\grid$ there is a $\elm_0\in\gridk[0]$ such
that $\elm\in \forest(\gridk[0])$. The generation $\gen(\elm)$
is the number of its ancestors in the tree
$\forest(\elm_0)$, or, equivalently, the number of bisections needed
to create $\elm$ from $\elm_0$. 

The following simple properties are useful.

\begin{lemma}\label{L:NVB}
  \begin{enumerate}[\upshape(1)]
  \item\label{h-gen} For $\elm\in\forest(\elm_0)$ with $\elm_0\in\gridk[0]$ we have
    \begin{displaymath}
      \hE = 2^{-\gen(\elm)/2} \hE[\elm_0].
    \end{displaymath}
  \item\label{elements_at_vertex} Defining
    \begin{math}
      \alpha_0 \definedas \max\{\#\gridz[z_0] \mid z_0 \in \verticesk[0]\}
    \end{math}
    we have for $d=2$ and $z\in\vertices$ the bound
    \begin{displaymath}
      \#\gridz \le
      \begin{cases}
        8& \text{if } z\in\vertices\setminus\verticesk[0],\\
        2\alpha_0 & \text{if } z\in\verticesk[0].
      \end{cases}
    \end{displaymath}
  \end{enumerate}
\end{lemma}
\begin{proof}
  Bisection halves the volume of a simplex. The 
  definition of $\gen(\elm)$ then gives the first claim. 
  During refinement any angle is bisected at most
  once, which yields the second assertion.
\end{proof}

The following assumption on a compatible distribution of refinement edges in
$\gridk[0]$ is instrumental in the analysis of \NVB, like the complexity
estimates in \cite{BiDaDe:04,Stevenson:08,2dnewest:13}.

\begin{assumption}[Compatibility Condition]\label{A:initial_grid}
  Suppose $\elm,\elm' \in \gridk[0]$ are direct neighbours with common edge
  $\elm\cap\elm'=\edge\in\edgesk[0]$. Then either $\edge$ is the common
  refinement edge of both $\elm$ and $\elm'$, or $\edge$ is the
  refinement edge of descendants $\elm''$ of $\elm$ and $\elm'''$ of $\elm'$ such that $\gen{\elm''} = \gen{\elm'''} < d$.
\end{assumption}

Mitchell has shown that a distribution of refinement edges, s.t. assumption \ref{A:initial_grid} holds, can be found for
any initial triangulation $\gridk[0]$ \cite[Theorem~2.9]{Mitchell:88} in 2 dimensions; compare also with
\cite[Lemma~2.1]{BiDaDe:04}.
The assumption particularly implies that any uniform
refinement of $\gridk[0]$ is conforming, \ie for any $g\in\N_0$ we
find that
\begin{math}
  \{\elm\in \forest(\gridk[0]) \mid \gen(\elm)=g\}\in\grids.
\end{math}
The proof of this property is a combination of \cite[\S4]{Traxler:97}
and \cite[Theorem~4.3]{Stevenson:08}. It is the key to show the following
property of \NVB; compare with \cite[Corollary 4.6]{Stevenson:08}.

\begin{proposition}[Characteristics of \NVB]\label{P:NVB}
  Suppose that the initial triangulation $\gridk[0]$ satisfies
  Assumption~\ref{A:initial_grid}. Let $\grid\in \grids$ be given and
  suppose that $\elm,\elm'\in \grid$ are direct neighbours such that the
  common edge $\edge=\elm\cap\elm'$ is the refinement edge of
  $\elm$. Then we either have
  $\gen(\elm')=\gen(\elm)$ and $\edge$ is also the refinement edge of
  $\elm'$, or $\gen(\elm')=\gen(\elm)-i$ with $1 \leq i < d$.
\end{proposition}

A simple consequence is
$\abs{\gen(\elm)-\gen(\elm')}\le d -1 $ for direct neighbours
$\elm,\elm'\in\grid$.

\section{Distributed Newest Vertex Bisection}

The extension of \NVB to the domain decomposition case is necessary, as the decomposed grid needs to be conforming across processor/partition boundaries. The basic idea is to execute the serial algorithm on each partition, communicate the refinement status of the partition boundary to the corresponding neighbour, refine conformingly and iterate.

In \cite{Zhang:08} it is shown that this parallel Algorithm \ref{algo:DNVB} yields the same final triangulation as the serial version for meshes that fulfill the compatibility condition \ref{A:initial_grid}. This is essentially due to the fact that there is a unique mapping from the set of marked elements to the final refinement situation.\\

\begin{algorithm}[H]
\caption{Distributed Newest Vertex Bisection}
\label{algo:DNVB}
Initialize set of elements marked for refinement on each Partition $M_i$, $0\leq i<P$.\\
\While{ $M_i \neq \emptyset \forall i$}{
	Refine Partition $P_i$ using \NVB until $M_i$ is empty\\
	Communicate refinement status of partition boundary to corresponding neighbour\\
	Add nonconforming simplices to $M_i$\\
	Communicate globally, whether $M_i$ is empty.	
	}
\end{algorithm}
\medskip

We improved the algorithm introduced in \cite{Zhang:08} by additionally communicating the edge status (i.e. whether the edge has been bisected) of edges belonging to the process boundaries. This is slightly more communication expensive in $3$ or more dimensions but it reduces the number of iterations needed. For 2 dimensions both algorithms coincide as faces are always 1 dimensional.\\
The stopping criteria of the while loop requires a global communication (Allreduce, $O(\log p)$ where $p$ is the number of partitions), which cannot be expected to scale well onto many cores.
On the other hand communicating the refinement status to the neighbour can be expected to scale quite well as long as the number of neighbouring partitions stays small, which relates to the quality of the decomposition.

\section{Communication of the Distributed Newest Vertex Bisection}

Bounds for the amount of communication necessary to reach a conforming triangulation are directly related to bounding the number of iterations in Algorithm \ref{algo:DNVB}. While the first bound from Theorem \ref{theorem:maxlevel} does not need more assumptions than Compatibility Condition \ref{A:initial_grid}, Theorem \ref{theorem:constant} additionally requires dimension $d=2$ and a certain form of mesh decomposition.\\
The following first lemma helps to understand the direct consequences of a single refinement.

\begin{lemma}
\label{lemma:directclosure}
For all direct neighbours $\elm'$ of an element $\elm \in \grid$ with refinement edge $\edge$ with $\edge \subset T'\cap T$ one of the following two statements holds 
\begin{itemize}
\item[1] Refinement of $\edge$ in $\elm'$ induces no further refinement (it already is the refinement edge)
\item[2] $\edge$ is the refinement edge of a child of $\elm'$ and refinement of $\edge$ induces up to $d-1$ refinements of elements $\elm_i \subset \elm'$ with $\gen(\elm_i) < \gen(\elm) $. 
\end{itemize}
\end{lemma}

\begin{proof}
Two cases:
\begin{itemize}
\item[1] $\gen(\elm) = \gen(\elm') \Rightarrow$ no further refinement. It is the refinement edge of both elements because of the compatibility condition.
\item[2] $\gen(\elm) = \gen(\elm') +i$ with $1 \leq i \leq d-1 \Rightarrow $ induces $i$ further refinements. The refinement edge can only be the same, if both elements have the same generation and direct neighbours can only differ in generation by $d$ at most.\\
The generation is increased by one with each refinement, so there have to be $i$ refinements of descendants of $\elm'$ that are refined, until the $i$-th descendant shares the refinement edge with $\elm$ and yields the conforming closure of that element.
\end{itemize} 
\end{proof}

Lemma \ref{lemma:directclosure} can be applied recursively. A single refinement may lead to refinements of direct neighbours at $1$ to $(d-1)$ generations lower and these may again lead to refinements at even lower generations.

\begin{example}\label{ex:inducedRT}
 Let us assume a simplex $\elm$ with generation $\gen(\elm)=0$ and its direct neighbour $\elm'$ with $\gen(\elm') = d-1$ and refinement edge $\edge = \elm \cap \elm'$. Now we refine $\elm'$, so we refine $\edge$ and hence $\elm$. The compatibility condition \ref{A:initial_grid} yields that $\edge$ is the refinement edge of a descendant $\elm''$ of $\elm$ with $\gen(\elm'')= d-1$. We denote by $\elm^i_\edge$ the descendant of $\elm$ with generation $i$ that contains $\edge$ and we denote its refinement edge by $\edge_i$. So $\elm'' = \elm^{d-1}_\edge$, $\elm = \elm^0_\edge$ and $\edge_{d-1} = \edge$. Then the Refinement Propagation can be depicted in the following graph of figure  \ref{fig:simpleIRT}.\\
\begin{figure}[htbp]
\begin{tikzpicture}[level distance=1.25cm]
\tikzstyle{edge from parent}=[draw=none]

\node (l3) {gen 3}
  child{ node (l2) {gen 2} 
    child{ node (l1) {gen 1} 
      child{ node (l0) {gen 0}
      }
    }
  };

\begin{scope}[xshift=20mm]
\node (el) {$\elm'$}
  child[grow=right]{node (edge) {$\edge$}
    child[grow=right,level distance=4cm]{node (elm2){$\elm''$}
      child[level distance=1.25cm,grow=down]{node (edge3){$\edge_2$}
        child[grow=down]{node (edge2) {$\edge_1$}
          child[grow=down]{node (edge1) {$\edge_0$}  }       
        }
      }
    }
  };

\draw[->,dashed] (el) -- (edge) -- (elm2);
\draw[->] (edge) --node[anchor=south,pos=.8]{$\elm^2_\edge$} (edge3);
\draw[->] (edge) --node[anchor=south,pos=.8]{$\elm^1_\edge$} (edge2);
\draw[->] (edge) --node[anchor=south,pos=.8]{$\elm$} (edge1);

\end{scope}

\end{tikzpicture}
\caption{The direct Refinement Propagation of example \ref{ex:inducedRT} with $d=4$. }
\label{fig:simpleIRT}
\end{figure}
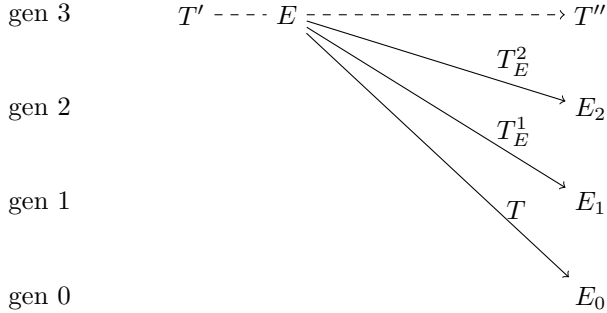
The dashed lines denote the direct closure. It cannot induce further refinement and can be neglected. More general we have to analyze the direct Refinement Propagation for any element that contains $\edge$ and additionally for all elements containing any of the edges $\edge_i$, as their refinement may also induce further refinement in the grid.
\end{example}

This leads us to the following definition.

\begin{definition}
\label{def:inducedRT}
Refinement of an element $\elm$ with refinement edge $e_0$ and generation $\gen(\elm)=l$ induces refinement propagation in form of a directed graph with root $e_0$. For any element $\elm'$ with $e_0 \subset \elm'$ and $\gen(\elm') =:l' < l$ we have a directed edge from $e_0$ to the refinement edge $\edge'_i$ of $\elm'^i_{e_0}$ for $l' \leq i < l$ as new nodes. We repeat by setting every newly introduced node as a local root.  \\
We call this the \emph{Refinement Propagation Graph}.
\end{definition}

\begin{example}
\begin{figure}[htbp]
\begin{tikzpicture}[level distance=1.25cm]
\tikzstyle{edge from parent}=[draw=none]

\Tree [.{gen l} [.{gen l-1} [.{gen l-2} {gen l-3} ] ] ]

\begin{scope}[xshift=50mm,sibling distance=1cm,
   edge from parent path={(\tikzparentnode) -- (\tikzchildnode)}]
\tikzstyle{edge from parent}=[draw=black,->]

\Tree 
[.\node (e0) {$e_0$}; 
    \edge node[auto=right] {$(T_1)^{l-1}_{e_0}$};
    [.\node (e1) {$e_1$};
    ]
    \edge[draw=none] node[auto=left] {};  
    [
       \edge[draw=none] node[auto=left] {};  
       [.\node (e2) {$e_2$};
       ]
    ]
    \edge node[auto=left] {$T_2$};
    [.\node (e3) {$e_3$};
    ]
    \edge node[auto=left] {$T_3$};
    [.\node (e4) {$e_4$};
       \edge node[auto=right] {$(T_4)^{l-2}_{e_4}$};
       [.\node (e5) {$e_5$};
       ]
          \edge[draw=none] node[auto=left] {};  
          [
             \edge[draw=none] node[auto=left] {};  
             [.\node (e6) {$e_6$};
             ]
          ]
       \edge node[auto=right] {$T_5$};
       [.\node (e7) {$e_7$};
       ]
       \edge node[auto=left] {$T_6$};
       [.\node (e8) {$e_8$};
       ]
    ]
];
\draw[->] (e0) -- (e2) node[midway,auto=right] {$T_1$};
\draw[->] (e4) -- (e6) node[midway,auto=right] {$T_4$};
\draw[->] (e3) -- (e5) node[midway,auto=left] {$T_7$};
\end{scope}

\end{tikzpicture}
\caption{An example for an Refinement Propagation Graph. }
\label{fig:IRT}
\end{figure}
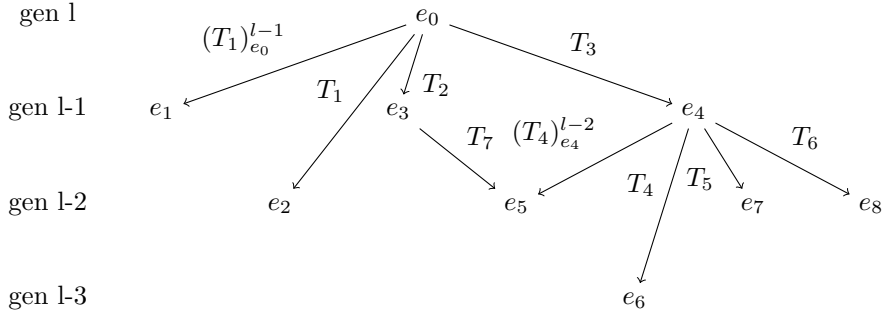

Figure \ref{fig:IRT} depicts an example of the Refinement Propagation graph of an initial refinement of an element $\elm$ with refinement edge $e_0$. For some elements the refinement results in the direct closure, so they are not included in the graph as well as $\elm$. For three direct neighbours $T_1,T_2,T_3$ of $\elm$ the refinement does not result in the direct closure. For $T_1$ it even results in an additional refinement of its child $(T_1)^{l-1}_{e_0}$ that contains $e_0$. Bisection of $e_4$ to refine $T_3$ is locally similar to the refinement of $\elm$ by $e_0$. Note that the graph is not a tree, as refinement edges may be shared ($e_5$ in our example).
All leafs do not induce further refinement, so all adjacent elements are of the same level and refinement is the direct closure, which is not included in the graph.
\end{example}

\begin{remark}
In 2 dimensions the Refinement Propagation Path consists solely of nodes of degree 2 (and the root and the leaf), since in 2 dimensions every edge is shared by exactly two elements and the direct closure is not included.  Hence in 2 dimensions we call it the \emph{Refinement Propagation Path}.
\end{remark}

With this preliminary work and exploiting the compatibility condition \ref{A:initial_grid} we state the following theorem.

\begin{theorem}\label{theorem:maxlevel}
Let $M \subset \grid$ be the set of elements marked for refinement. Then the number $N$ of iterations in the Algorithm \ref{algo:DNVB} to reach a conforming state satisfies
\[
N \leq \max_{\elm \in M}\max_{\elm'\in \grid} (\gen(\elm)-\gen(\elm')) +2.
\]
\end{theorem}

\begin{proof}
Let $\elm \in M$ with $\gen(\elm)=l$. We will bound the maximum depth of the Refinement Propagation Graph of $\elm$, which is an upper bound for the number of iterations as in the worst-case scenario refinement needs to be communicated at every edge.\\
Due to Lemma \ref{lemma:directclosure} refinement can be propagated at generation $l-d < \gen(\elm') < l$, in particular at generation $l-1$ and no propagation at generation $l$. So the maximum number of propagations $N_\elm$ resulting from refining $\elm$ is $l - \min_{\elm' \in \grid} \gen(\elm')$. This is clearly also the maximum depth of the Refinement Propagation graph.\\
We have to take into account, that the Refinement Propagation graph does not consider the direct closure, which could need an additional communication. It follows
\[
N_\elm \leq l - \min_{\elm' \in \grid} \gen(\elm') +1
\]
If we now take the maximum over all $T \in M$ this results in \[\max_{\elm \in M}\max_{\elm'\in \grid} (l(\elm)-l(\elm')) +1\,.\]  \\
We have to add another $1$ as Algorithm \ref{algo:DNVB} has to communicate that it has finished.
\end{proof}

The following example demonstrates that this bound is sharp and that we cannot expect anything 
better even in the simple case of 2 partitions. In particular the 
bound $O(\log p)$ from \cite[Sec. 2.4]{amdis:15} cannot hold 
without further assumptions on the decomposition.

\begin{example}
\label{ex:worstcase}
\begin{figure}[htbp]

\begin{tikzpicture}
 \path[coordinate] (0,1.5)  coordinate(A)
                --( 5,3) coordinate(B)
                --(5,0) coordinate(C)
                --(10,1.5) coordinate(D)
                --(10,4.5) coordinate(E);
        \path[coordinate] (C) -- (B) coordinate[pos=.5](CB);
    \draw [fill=white,thick](A) -- (B) -- (C) -- cycle;
    \draw [fill=white,thick](D) -- (B) -- (C) -- cycle;
    \draw [fill=green!10,thick](B) -- (E) -- (D) -- cycle;
    \draw [blue,dashed] (C) -- (E) coordinate[pos=.5] (CE);

        \path[coordinate] coordinate(X) at (B){};
        \path[coordinate] (C) -- (X) coordinate[pos=.5](CX);                          
        \draw[blue,dashed] (C) -- (CB) -- (CX);
        \draw[blue,dashed] (CE) -- (CB);
        \path[coordinate] (A) -- (B) coordinate[pos=.5](B);
        \path[coordinate] (C) -- (B) coordinate[pos=.5](CB);
        \path[coordinate] (C) -- (X) coordinate[pos=.5](CX);        
        \draw[fill=green!10,dotted] (C)--(B)--(X)--cycle;
        \draw[blue,dashed] (B) -- (CX) -- (CB);      
        \path[coordinate] (A) -- (C) coordinate[pos=.5](C);        
        \draw[dotted] (A)--(B)--(C)--cycle;

    \foreach \x in {2,3}{%
    
        \path[coordinate] coordinate(X) at (B){};
        \path[coordinate] (C) -- (X) coordinate[pos=.5](CX);                          
        \draw[blue,dashed] (C) -- (CB) -- (CX);
        \path[coordinate] (A) -- (B) coordinate[pos=.5](B);
        \path[coordinate] (C) -- (B) coordinate[pos=.5](CB);
        \path[coordinate] (C) -- (X) coordinate[pos=.5](CX);        
        \draw[fill=green!10,dotted] (C)--(B)--(X)--cycle;
        \draw[blue,dashed] (B) -- (CX) -- (CB);      
        \path[coordinate] (A) -- (C) coordinate[pos=.5](C);        
        \draw[dotted] (A)--(B)--(C)--cycle;
}
\draw[red,thick] (C) -- (CB);
\end{tikzpicture}

\caption{A distributed refined mesh to illustrate that the bound of theorem \ref{theorem:maxlevel} is sharp.
Partitions are indicated by the color of the cells. Black lines denote macro element borders. Dotted lines denote the initial refinement situation. The refinement request is marked in red. Blue dashed lines denote its Refinement Propagation.}
\label{fig:worstcase}
\end{figure}
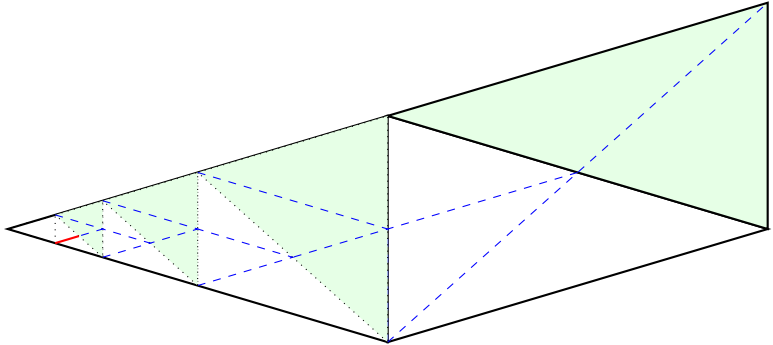
Figure \ref{fig:worstcase} shows a mesh consisting of three initial triangles fulfilling the compatibility condition \ref{A:initial_grid}.  One triangle has been refined into a corner, so the minimum generation in the mesh is $\min_{\elm'\in\grid}\gen(\elm') =0$ and the generation of the marked element $\elm \in M = \{\elm\}$ is $\gen(\elm)=6$, so the bound from Theorem \ref{theorem:maxlevel}   is $8$ iterations. Every refinement induces additional refinement at exactly one generation lower, so the refinement propagation path traverses $7$ edges. The mesh is partitioned in such a way that every edge lies on a processor boundary, which implies that after every refinement Algorithm \ref{algo:DNVB} has to stop and globally communicate. This means we get $7$ iterations, where the marked set is not empty on all processors and one final iteration to communicate, that we are finished. In total this is $8$ iterations, which is the bound predicted from Theorem \ref{theorem:maxlevel}.  \\
Distributing the elements into partitions as depicted in figure \ref{fig:worstcase} is not purely artificial. Sorting the leaf elements in vertical direction with respect to their center coordinates leads to these partitions.
\end{example}

Example \ref{ex:worstcase} demonstrates that we have to impose additional assumptions on the decomposition to expect better bounds on the number of iterations of Algorithm \ref{algo:DNVB}.\\ A reasonable assumption could be that the partitions are created by a hierarchical space-filling curve, such that elements that are close in the refinement-tree are probably on the same partition.\\
Another assumption is that the mesh is partitioned solely on elements of a generation $l^*$ which are then distributed onto partitions together with their respective refinement trees. This second assumption will be analyzed in this paper, because this is the partitioning currently implemented in our grid manager \dune[ALUGrid]. Without loss of generality we can assume that we distribute elements on the macro level (i.e. elements $\elm$ with $\gen(\elm)=0$) with their respective refinement trees. If this was not the case we could set the uniformly refined grid as our new initial grid, as for this kind of partitioning coarsening below generation $l^*$ is forbidden.\\
The following results hold for partitioning on any fixed level but only for $2$-dimensional grids, as we rely on the fact that we have a Refinement Propagation Path instead of a full graph.

\begin{theorem}
\label{theorem:constant}
Let dimension $d=2$ and the mesh be partitioned as described above. Let $z \in \vertices$  and let $N_z = \{\elm \in \gridk[0] : z \subset \elm \}$ the set of macro elements containing that vertex. For an element $\elm$ in the set of marked elements $M$ let $\elm_0(\elm) \supset \elm$ be the element of the macro grid $\gridk[0]$ that contains $\elm$. Then the number of global communications $N$ in Algorithm \ref{algo:DNVB} satisfies
\[
N \leq \max_{\elm \in M}\max_{z \in \elm_0(\elm)} \frac{3}{4} \# N_z + \frac{7}{4} \leq  \max_{z \in \gridk[0]} \frac{3}{4} \# N_z +\frac{7}{4}.
\]
\end{theorem}

\begin{proof}
The second inequality is trivial. The proof of the first inequality splits into two parts. 
\begin{enumerate}
\item Refinement propagation around vertex $z$.
\item Refinement propagation inside of macro elements that contain vertex $z$.
\end{enumerate}
We start with a similar observation as in the proof of Theorem \ref{theorem:maxlevel}. By bounding the traversals of edges of $\gridk[0]$ within the refinement propagation path, we bound the number of global communications. We count edges of the refinement propagation path that are contained in edges of $\gridk[0]$.\\
Let $\elm \in M$ with $z \in \elm_0(\elm)$ be the element to be refined. Refinement of $\elm$ leads to elements of generation $l=\gen(\elm)+1$. All refinement propagation of refinement of $\elm$ is contained in the refinement propagation of uniformly refining $\elm_0$ to this level, which we will investigate. (cf. Figure \ref{fig:spider}) \\
(1): Refinement propagation around vertex $z$.\\ We are now investigating the effect of refinement of $\elm_0$ at its vertex $z$. There are two possibilities:
\begin{itemize}
\item[a.] $z$ is opposite of the initial refinement edge of $\elm_0$. \\Then there are two leaf elements $\elm^{0,1}$ with $z \in \elm^{0,1} \subset \elm_0$ and $\gen(\elm^0)=\gen(\elm^1)$. If $\gen(\elm^{0,1})/2 \mod 2=0$, the refinement edge of $\elm^0$ and $\elm^1$ is the shared edge. If $\gen(\elm^{0,1})/2 \mod 2=1$ the refinement edge of $\elm^{0,1}$ is a subedge of an edge of $\elm_0$ containing $z$ and as we are uniformly refining up to level $l$, every odd level there is refinement propagation across these edges.
\item[b.] $z$ is contained in the initial refinement edge of $\elm_0$. \\Then there is one leaf element with $z \in \elm' \subset \elm_0$.  The refinement edge of $\elm'$ is always contained in one of the edges of $\elm_0$. So every level refinement of $\elm'$ leads to refinement propagation across one of the two edges of $\elm_0$ containing $z$.
\end{itemize} 
We know that elements around the vertex $z$ form the refinement propagation path of $\elm'$, as they all differ by one in generation and due to the argument above. The path evolves around the vertex until it encounters an element with the lowest level.  The maximum level difference around the vertex is $(\#\gridz -1) /2$ as there have to exist at least two elements with lowest generation. So the maximum length of the refinement propagation path around $z$ is $(\#\gridz -1) /2$. We want to bound the refinement propagation path edge traversals with respect to $\# N_z$. $\# N_z$ is smaller than $\#\gridz$, as for every element in $\# N_z$ where the refinement edge is opposite of $z$, there are two elements in $\#\gridz$. Now there are two cases:
\begin{enumerate}
\item[a.] $\# N_z$ is even:\\
The worst case is $\#\gridz = 3/2 \# N_z$ and all refined macro angles are neighbouring. Then, in one of the circumvention direction all edge traversals are traversals in $N_z$ and so we may need $(3/2 \# N_z -1)/2$ traversals in $N_z$ until we encounter the element with the lowest level. 
\item[b.] $\# N_z$ is odd:\\
The worst case is $\#\gridz = 3/2\# N_z+1/2 $ and all refined angles are neighbouring. Then, in one of the circumvention direction all edge traversals are traversals in $N_z$ and so we may need $(3/2 (\# N_z )+1/2 -1)/2$ traversals in $N_z$ until we encounter the element with the lowest level.
\end{enumerate}
So the number of macro edge traversals $N^e_{z}$ within the refinement propagation path around $z$ satisfies
\[
N^e_{z} \leq (3/2 (\# N_z - 1/2)/2 = 3/4 \# N_z -1/4.
\]
(2): Refinement propagation inside of macro elements that contain vertex $z$.\\
We already know that the uniform refinement propagates into all macro elements, which share $z$. So we can neglect macro elements that share an edge with $\elm_0$ as we know that their refinement propagation does not yield any additional information.\\
So we consider a macro element $\elm_1$ that contains $z$ and does not share an edge with $\elm_0$. We know that we have a refinement at vertex $z$ and want to investigate, whether we can reach the edge opposite of $z$. The other edge does not matter, as it contains $z$. \\
The size of an element of generation $l$ is $2^{-l}|\elm_0|$ (Lemma \ref{L:NVB}). The refinement propagation path consists of elements that differ by one in level. So the size of the refinement propagation path inside the element $\elm_1$ is $\sum_k 2^{-k} |\elm_1| = |\elm_1|(2^{-n} - 2^{-l})$. So to get to another element from the opposite vertex it has to include an element of level $0$, a macro element. This is only possible, if the initial refinement edge is opposite of $z$.\\
In combination with the previous result this yields
\[
N \leq (3/2\# N_z - 1/2)/2 +1 = 3/4 \# N_z + 3/4.
\]
Now we finish the proof by adding another $+1$ for communicating the final status and taking the maximum  over all vertices of $\elm_0$ and over all marked elements.
\end{proof}

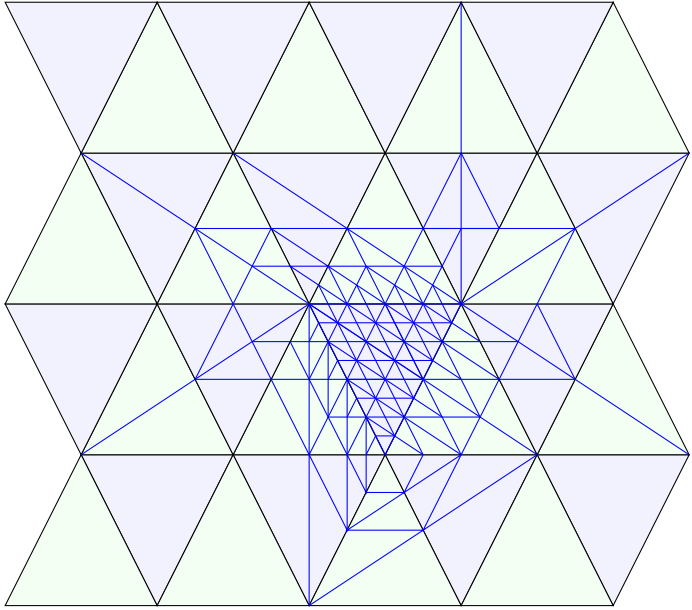
\begin{figure}[htbp]
\begin{tikzpicture}[scale=2]  

\foreach \x in {0,1,...,4}{%
\path[coordinate] (\x,0) coordinate (A\x) -- (\x+0.5,1) coordinate(B\x) -- (\x,2) coordinate (C\x) -- (\x+0.5,3)coordinate(D\x) -- (\x,4)coordinate(E\x);
}
\newcount\oldx
\foreach \x in {1,...,4}{%
\oldx =\x 
\advance\oldx by  -1
\draw[fill=blue!5] (B\the\oldx) --  (B\x) -- (A\x) --   cycle ;
\draw[fill=green!5] (A\the\oldx) -- (A\x) -- (B\the\oldx) -- cycle;
\draw[fill=green!5] (B\x) -- (C\x) -- (B\the\oldx) -- cycle;
\draw[fill=blue!5] (C\the\oldx) -- (C\x) -- (B\the\oldx) -- cycle;
\draw[fill=blue!5] (D\x) -- (C\x) -- (D\the\oldx) -- cycle ;
\draw[fill=green!5] (D\the\oldx) -- (C\x) -- (C\the\oldx) -- cycle;
\draw[fill=green!5] (D\x) -- (E\x) -- (D\the\oldx) -- cycle;
\draw[fill=blue!5] (E\the\oldx) -- (E\x) -- (D\the\oldx) -- cycle;

}

\foreach \f in {0,0.25,0.5,0.75,1}{%
	\path[coordinate] (C3) -- (C2)coordinate[pos=\f](X);
	\path[coordinate] (B2) -- (C2)coordinate[pos=\f](Y);	
	\path[coordinate] (C3) -- (B2)coordinate[pos=\f](Z);
	\path[coordinate] (B2) -- (C3)coordinate[pos=\f](ZU);

	\path[coordinate] (C3) -- (X) coordinate[pos=.5](X1);
	\path[coordinate] (C2) -- (X)coordinate[pos=.5](X2);	
	\path[coordinate] (C2) -- (Y) coordinate[pos=.5](Y1);
	\path[coordinate] (B2) -- (Y)coordinate[pos=.5](Y2);
	\path[coordinate] (C3) -- (Z) coordinate[pos=.5](Z1);
	\path[coordinate] (B2) -- (Z)coordinate[pos=.5](Z2);	
	\path[coordinate] (C3) -- (ZU) coordinate[pos=.5](ZU1);
	\path[coordinate] (B2) -- (ZU)coordinate[pos=.5](ZU2);

	\draw[blue](X) --(Y);
	\draw[blue](Y) --(ZU2);
	\draw[blue](X) --(Z1);
	\draw[blue](X1) -- (Z1);
	\draw[blue](X2) -- (Z2);
	\draw[blue](Y1) -- (ZU1);
	\draw[blue](Y2) -- (ZU2);	
}

\newcounter{curr}{1}
\setcounter{curr}{1}
\foreach \f in { 0.125, 0.25, 0.375,0.5,0.625,0.75,0.875}{%
\path[coordinate] (C3)--(C2)coordinate[pos=\f](X\thecurr) -- (B2)coordinate[pos=\f](Y\thecurr) -- (C3)coordinate[pos=\f](Z\thecurr);
\stepcounter{curr}
}

\setcounter{curr}{0}
\foreach \f in { 0.125, 0.25, 0.5, 0.75, 0.875}{%
\path[coordinate] (C3)--(B3)coordinate[pos=\f](CB3\thecurr) -- (B2)coordinate[pos=\f](BB3\thecurr) -- (B1)coordinate[pos=\f](BB2\thecurr) -- (C2)coordinate[pos=\f](BC2\thecurr) -- (D2)coordinate[pos=\f](CD3\thecurr) -- (C3)coordinate[pos=\f](DC3\thecurr);

\path[coordinate] (B3)--(C4)coordinate[pos=\f](BC4\thecurr) -- (C3)coordinate[pos=\f](CC4\thecurr) -- (D3)coordinate[pos=\f](CD4\thecurr) -- (D2)coordinate[pos=\f](DD3\thecurr) -- (D1)coordinate[pos=\f](DD2\thecurr) -- (C2)coordinate[pos=\f](DC2\thecurr) -- (C1)coordinate[pos=\f](CC2\thecurr) -- (B1)coordinate[pos=\f](CB1\thecurr) --(A2)coordinate[pos=\f](BA2\thecurr) -- (B2)coordinate[pos=\f](AB2\thecurr) -- (A3)coordinate[pos=\f](BA3\thecurr) -- (B3)coordinate[pos=\f](AB3\thecurr);
\stepcounter{curr}
}

\draw[blue] (Z4) -- (B3)coordinate[pos=.5](F1)coordinate[pos=.25](F2);
\draw[blue] (CB32) -- (Z4)coordinate[pos=.5](F3) -- (BB32)coordinate[pos=.5](F4);
\draw[blue] (CB32) -- (Z6);
\draw[blue] (BB32) -- (Z2);

\draw[blue] (CB31) -- (Z6) -- (F3) --(F1) -- (F4) -- (Z2) -- (BB33);
\draw[blue] (CB32) -- (F1) --(BB32);

\draw[blue] (Z7) -- (CB31);
\draw[blue] (Z5) -- (F3);
\draw[blue] (Z3) -- (F4);
\draw[blue] (Z1) -- (BB33);

\draw[blue] (C3) -- (CD32)coordinate[pos=.5](G1)coordinate[pos=.25](G2);
\draw[blue] (X4) -- (CD32)coordinate[pos=.5](G3) -- (DC32)coordinate[pos=.5](G4);

\draw[blue] (CD31) -- (X4)coordinate[pos=.5](G5) -- (G1)coordinate[pos=.5](G6) -- ( DC32);

\draw[blue] (CD31) -- (X6);
\draw[blue] (G1) -- (X2);
\draw[blue] (CD31) -- (G3) -- (G1) -- (DC33);

\draw[blue] (CD30) -- (X6) -- (G5) -- (G3) -- (G6) -- (X2) -- (G2) -- (DC33);

\draw[blue] (X7) -- (CD30);
\draw[blue] (X5) -- (G5);
\draw[blue] (X3) -- (G6);
\draw[blue] (X1) -- (G2);

\draw[blue] (C2) -- (BB22)coordinate[pos=.5](H1)coordinate[pos=.25](H2);
\draw[blue] (Y4) -- (BB22)coordinate[pos=.5](H3) -- (BC22)coordinate[pos=.5](H4);

\draw[blue] (BB21) -- (Y4)coordinate[pos=.5](H5) -- (H1)coordinate[pos=.5](H6) -- ( BC22);

\draw[blue] (BB21) -- (Y6);
\draw[blue] (H1) -- (Y2);
\draw[blue] (BB21) -- (H3) -- (H1) -- (BC23);

\draw[blue] (BB20) -- (Y6) -- (H5) -- (H3) -- (H6) -- (Y2) -- (H2) -- (BC23);

\draw[blue] (Y7) -- (BB20);
\draw[blue] (Y5) -- (H5);
\draw[blue] (Y3) -- (H6);
\draw[blue] (Y1) -- (H2);

\draw[blue] (C3) -- (BC42)coordinate[pos=.5](I1);
\draw[blue] (CC42) -- (BC42) -- (CB32);
\draw[blue] (CC42) -- (I1) -- (CB32);
\draw[blue] (CB31) -- (I1);

\draw[blue] (B3) -- (BA32)  -- (BB32) -- (BA31) -- (BB33);

\draw[blue] (C3) -- (DD32)coordinate[pos=.5](K1);
\draw[blue] (CD42) -- (DD32) -- (DC32);
\draw[blue] (CD42) -- (K1) -- (DC32);
\draw[blue] (DC33) -- (K1);

\draw[blue] (D1) -- (CD32) -- (DC22) --   (CD31) -- (DC23) -- (CD30);

\draw[blue] (C2) -- (CB12)coordinate[pos=.5](L1);
\draw[blue] (BC22) -- (CB12) -- (CC22);
\draw[blue] (BC22) -- (L1) -- (CC22);
\draw[blue] (BC23) -- (L1);

\draw[blue] (A2) -- (BB22) -- (AB22) -- (BB21) -- (AB23) -- (BB20);

\path[coordinate] (D3) -- (C4)coordinate[pos=.5](DC42);
\draw[blue] (C3) -- (DC42);
\draw[blue] (CC42) -- (DC42) -- (CD32);
\draw[blue] (DC42) -- (D4);
\draw[blue] (DD32) -- (E3);

\draw[blue] (BC42) -- (B4);

\path[coordinate] (C1) -- (D1)coordinate[pos=.5](CD12);
\draw[blue] (C2) -- (CD12)coordinate[pos=.5](M1);
\draw[blue] (CC22) -- (CD12) -- (DC22);
\draw[blue] (CC22) -- (M1) -- (DC22);
\draw[blue] (M1) -- (DC23);

\draw[blue] (CD12) -- (D0);

\draw[blue] (CB12) -- (B0);

\draw[blue] (A2) -- (BA32) -- (AB22) -- (BA31) -- (AB23);

\end{tikzpicture}
\caption{Refinement propagation from uniform refinement of a single macro element.}
\label{fig:spider}
\end{figure}

\begin{remark}
From part (2) of the proof one can see, that if the mesh is uniformly refined up to one level below the  macro level, then the $+1$ from this part disappears and we get.
\[
N \leq \max_{\elm \in M}\max_{z \in \elm_0(\elm)} 3/4 \# N_z +3/4 \leq  \max_{z \in \gridk[0]} 3/4 \# N_z +3/4
\]
\end{remark}

\begin{remark}
\label{rem:P}
If we do not count all edges of $\gridk[0]$, but only those that are actually processor boundaries, we get the following bound that is better as long as every processor gets "nice" partitions.
\[
N \leq \max_{\elm \in M}\max_{z \in \elm_0(\elm)} \# P_z - 1 +1 = \max_{\elm \in M}\max_{z \in \elm_0(\elm)} \# P_z
\]
where $\# P_z$ is now the number of partitions that share a vertex. The proof is similar to the proof of theorem \ref{theorem:constant}, but we get $\# P_z -1$ as we cannot argue with circumventions in both directions and the $+1$ is again due to the final communication. Note that if a processors partition $p$ has several elements containing $z$ and there is no path of elements $\elm \in p$ containing  $z$ connecting two elements, every connected subdomain has to be counted as a partition.
\end{remark}

\begin{remark}
Theorem \ref{theorem:constant} provides a mesh constant. So no dependence on number of processors is required, but just a regular initial mesh (see $\alpha_0$ in Lemma \ref{L:NVB}).
\end{remark}

\begin{remark}
We believe that theorem \ref{theorem:constant} holds in a similar way for higher dimensions. It is a hard problem as the shape of the Refinement Propagation graph is not known.
\end{remark}

Based on the estimate in remark \ref{rem:P} we propose a new improved algorithm for compatible meshes.

\begin{algorithm}[H]
\caption{Improved Distributed Newest Vertex Bisection}
\label{algo:ImprDNVB}
Initialize set of elements marked for refinement on each Partition $M_i$, $0\leq i<P$.\\
\For{$i = 0, \ldots, max_{z \in \gridk[0]} \#P_z - 1$}{
	Refine Partition $P_i$ using \NVB until $M_i$ is empty\\
	Communicate refinement status of partition boundary to corresponding neighbour\\
	Add nonconforming simplices to $M_i$\\
	}
\end{algorithm}

For compatible meshes we have proven, that this algorithm yields the conforming closure and reaches the final status. So communicating the final status is no longer necessary. Hence we take the better bound $\max_{z \in \gridk[0]}\#P_z -1$ instead of $\max_{z \in \gridk[0]}\#P_z$. We can directly use the bound from Theorem \ref{theorem:constant} and get an algorithm, which does not need global communication at all. Unfortunately we cannot expect meshes to be compatible, especially in 3 dimensions. In this case 
we propose to use a mixture of both algorithms, where a fixed number of loops 
is done like in \ref{algo:ImprDNVB} before switching to the 
Algorithm \ref{algo:DNVB} after a global communication, whether $M_i \neq \emptyset \  \forall i$.


\section{Implementation}

The \NVB algorithm is implemented in the open-source 
package \dune[ALUGrid] available 
at \url{https://gitlab.dune-project.org/extensions/dune-alugrid}.
\cite{alugrid:16} contains a description of the software and various examples. 

In this section we discuss the implementation of the 
communication procedures which is not contained in 
detail in \cite{alugrid:16}.
Communication is needed to 
exchange the refinement flags, e.g. line 4 in
Algorithm \ref{algo:DNVB} and \ref{algo:ImprDNVB}. 
It is essential for the
Distributed \NVB that this is done in a very efficient way to guarantee excellent
scalability. In \dune[ALUGrid] we have chosen to interleave the send and receive procedures 
with the packing and unpacking of refinement information. This way we are able to hide
some of the communication latency behind the necessary pack and unpack of information. 
We briefly sketch the send and pack routine as well as the receive and
unpack routine used in \dune[ALUGrid]. \\

Let $\mathcal{L}^s_p$ be the set of all ranks that process $p \in [0,P-1]$ sends data
to and $\mathcal{L}^r_p$ the corresponding set $p$ received messages from. 
We call the communication symmetric if $\mathcal{L}^s_p = \mathcal{L}^r_p$.
Asymmetric communication occurs, for example, during the load balancing, where the list of
send and receive ranks can differ. The communication algorithm, however, is the same. 
Using the sets $\mathcal{L}^s_p, \mathcal{L}^r_p$ the
corresponding methods for \textit{pack-and-send} and \textit{receive-and-unpack} are
briefly explained in Algorithm \ref{alg:packandsend} and \ref{alg:receiveandunpack},
respectively. In the following $\mathcal{T}^q_p$ denotes the set of simplices on process $p$ with linkage
to rank $q \in \mathcal{L}^{s,r}_p$.
\begin{algorithm}[!ht]
  \caption{Pack and send}
  \label{alg:packandsend}
  \For{ $q \in \mathcal{L}^s_p$ }{
    \For{$T \in \mathcal{T}^q_p$}{
       pack refinement information for simplex $T$ and communication link $q$\\
     }
     post non-blocking \textbf{MPI\_Isend} for communication link $q$\\
   }
\end{algorithm}

\begin{algorithm}[!ht]
 \caption{Receive and unpack}
 \label{alg:receiveandunpack}
 \For{ $q \in \mathcal{L}^r_p$ }{
    $r_q \leftarrow 0$ \\
 }
 $n_r \leftarrow 0$ \\
 \While{ $n_r < |\mathcal{L}^r_p|$ }
 {
   \For{ $q \in \mathcal{L}^r_p$ }{
     \If{ $r_q = 0$ }{
      $r_q \leftarrow $ \textbf{MPI\_Iprobe}( $q$ ) \\
       \If{ $r_q = 1$ }{
         $s \leftarrow $ \textbf{MPI\_Getcount}( $q$ ) \\
         resize buffer for received message size $s$\\ 
         post \textbf{MPI\_Recv}( $q$ ) to write message from $q$ to buffer \\
         \For{$T \in \mathcal{T}^q_p$}{
         \textbf{unpackData}( $T$, $q$ ) 
        }
        $n_r \leftarrow n_r + 1$ 
       }
     }
   }
 }
 \textbf{MPI\_Waitall}( $\mathcal{L}^s_p$ ), see Algorithm \ref{alg:packandsend} 
\end{algorithm}

Both algorithms are implemented in \dune[ALUGrid] and work for very general data sets
and are thus used for all point to point communications in the package. 

\section{Numerical Experiments}

In this section we show for various examples that the theoretical results can be
reproduced and that very good scalability for the adaptation algorithm is observed on
a petascale super computer.  

\subsection{Verification of theoretical results}
The first experiment aims to reflect the theoretical results as we 
construct a worst-case experiment. A mesh is partitioned such that 
each process gets exactly one macro element. Then we refine a single element at 
one of its vertices up to level 20 and we examine the number of iterations 
necessary to reach the conforming status.

\begin{figure}[!ht]
\includegraphics[width=0.4\linewidth]{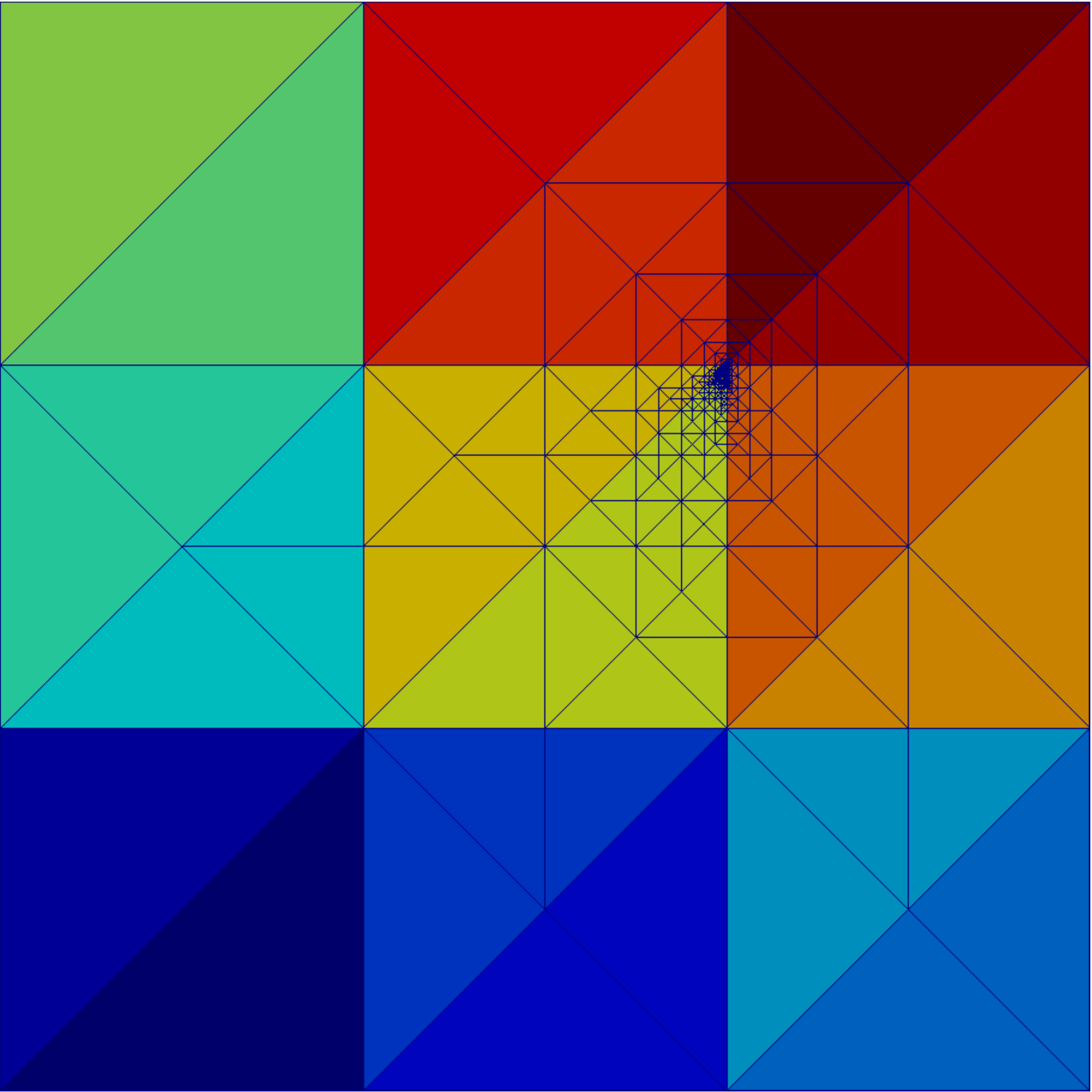} \hfil \begin{tikzpicture}[scale=0.6]
	\begin{axis}[
		xlabel=Max Level,
		ylabel=Refinement Loops,
		ymax=7
	]
	\addplot file{extractCompGrid18procs.txt};
	\addplot[red,sharp plot,update limits=false]  coordinates{(-5, -4) (30,31)};	
	\addplot[red,sharp plot,update limits=false]  coordinates{(-4, 6.25) (32,6.25)};		
	\end{axis}
\end{tikzpicture} 
\caption{A 2d unit square with 18 macro elements on 18 processors. The central yellow element gets refined at the top vertex. The number of macro neighbours $\# N_z = 6$. In the right plot red lines denote the two bounds from the Theorems \ref{theorem:constant} and \ref{theorem:maxlevel}.}
\label{fig:2dCompExp}
\end{figure}
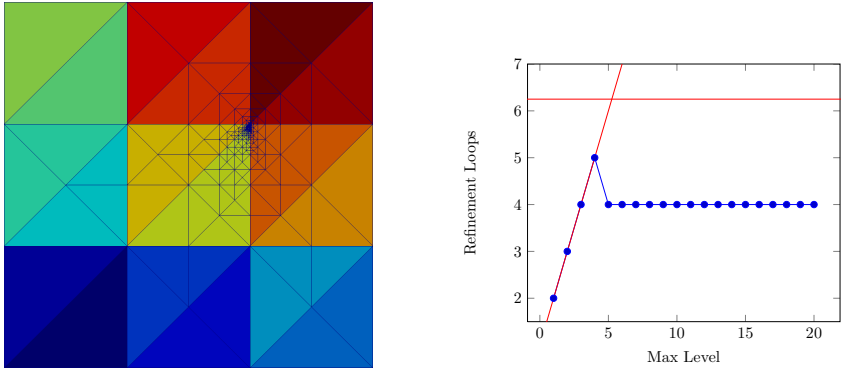

From figure \ref{fig:2dCompExp} we see that for compatible 2d grids both bounds(red lines) are not violated by the current implementation. As this is a worst-case scenario, in the average case the behaviour is better. Now we perform the same experiment on a non-compatible grid.

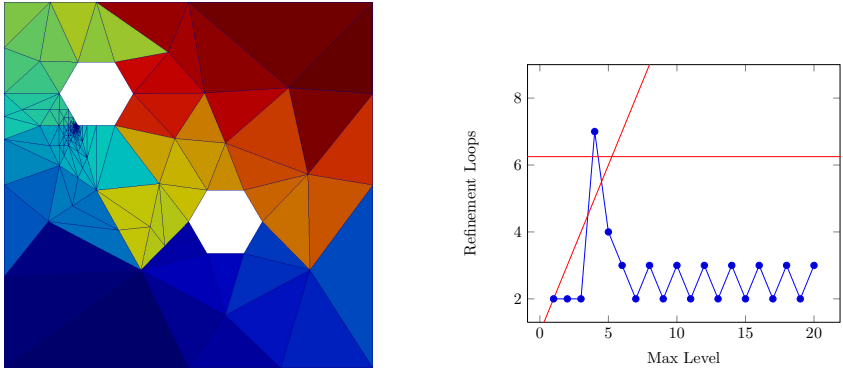
\begin{figure}[!ht]
\includegraphics[width=0.4\linewidth]{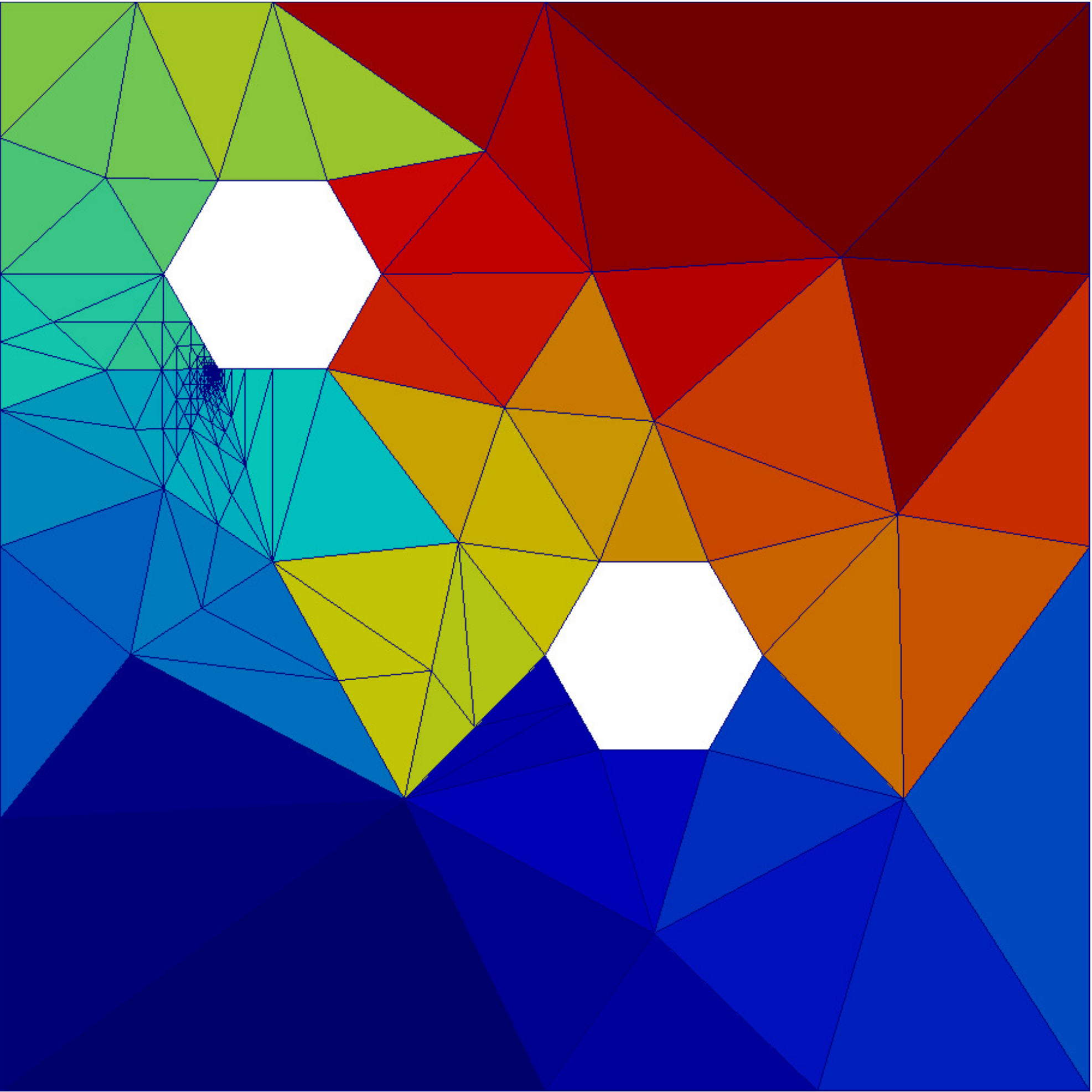} \hfil \begin{tikzpicture}[scale=0.6]
	\begin{axis}[
		xlabel=Max Level,
		ylabel=Refinement Loops,
		ymax=9
	]
	\addplot file{extractnonCompGrid60procs.txt};
	\addplot[red,sharp plot,update limits=false]  coordinates{(-5, -4) (30,31)};	
	\addplot[red,sharp plot,update limits=false]  coordinates{(-4, 6.25) (32,6.25)};		
	\end{axis}
\end{tikzpicture} 
\caption{A non-compatible 2d grid with 60 macro elements on 60 processors. An element (with $\# N_z = 6$) gets refined at a vertex. In the right plot red lines denote the two bounds from the Theorems \ref{theorem:constant} and \ref{theorem:maxlevel}.}
\label{fig:2dNonCompExp}
\end{figure}

As expected, the plot in Figure \ref{fig:2dNonCompExp} illustrates that compatibility is a necessary condition for both theorems. This means although the implementation is capable of handling non-compatible 2d grids, the bounds from the theory do not hold.

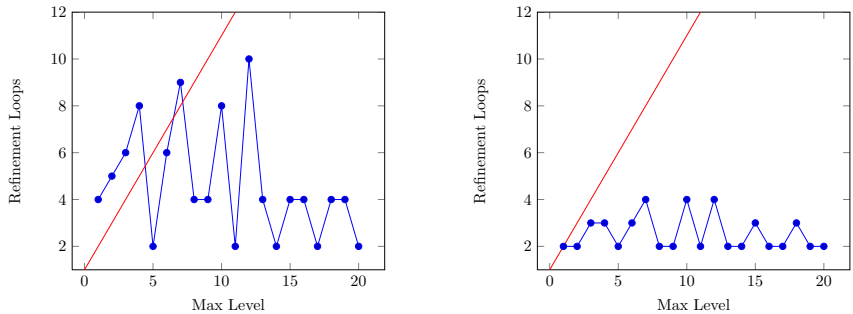
\begin{figure}[!ht]
\begin{tikzpicture}[scale=0.6]
	\begin{axis}[
		xlabel=Max Level,
		ylabel=Refinement Loops,
		ymax=12
	]
	\addplot file{extractCompGrid162procs.txt};
	\addplot[red,sharp plot,update limits=false]  coordinates{(-5, -4) (30,31)};	
	\end{axis}
\end{tikzpicture} 
\hfil
\begin{tikzpicture}[scale=0.6]
	\begin{axis}[
		xlabel=Max Level,
		ylabel=Refinement Loops,
		ymax=12
	]
	\addplot file{extractCompGridEdges162procs.txt};
	\addplot[red,sharp plot,update limits=false]  coordinates{(-5, -4) (30,31)};	
	\end{axis}
\end{tikzpicture} 
\caption{Unit cube. 162 Macro Elements/ 162 processors ($3 \times 3 \times 3$ Kuhn-dice). Central element gets refined at one vertex. On the left refinement status is communicated on faces. On the right refinement status is first communicated over edges and then additionally over faces. }
\label{fig:3dCompExp}
\end{figure}

We expect the bound from Theorem \ref{theorem:maxlevel} to hold in any dimension. This is not the case in figure \ref{fig:3dCompExp}. This failure arises from an implementation detail, that refinement status 
is communicated for faces instead of edges. 
After implementation of an additional communication of edges statuses during refinement 
we see that the theoretical result holds. In addition, the plot indicates 
the existence of a constant to bound the number of 
iterations in 3d similar to Theorem \ref{theorem:constant}.

\subsection{Strong scaling experiments}

In Figure \ref{fig:2dball} we show the refinement of a doughnut that rotates
around the center (doughnut refinement). Triangles inside the doughnut are refined and
triangles outside are coarsened. This test was introduced in \cite{alugrid:16} and 
serves as an excellent test 
for the Distributed \NVB since frequent refinement and
coarsening occurs throughout the simulation. In fact, the adaptation and load balancing
is performed in each time step. The domain decomposition is based on the Hilbert space
filling curve approach implemented in Zoltan \cite{zoltan}

Figures \ref{fig:2dscaling} and \ref{fig:3dscaling} we provide strong scaling results obtained for the 
doughnut refinement test in 2d and 3d, respectively. The scaling experiment have been
performed on the super computer Yellowstone \cite{Yellowstone}.
The observed strong scaling for the adaptation cycle (green triangular line) 
is excellent for all experiments and very close to the optimal scaling. The load
balancing (pink boxed line) on the other hand at some points fails to scale well because 
the number of macro simplices per core becomes to small as the number of processes
grows which results from the drawback of basing the partitioning upon the macro mesh. 
This is currently under investigation and will be reported in a separate article. 
In contrast to \cite[Fig. 8]{amdis:15}, where for a different adaptation experiment 
the \textit{mesh adaptation loop} yielded non-optimal strong scaling, 
we conclude from our investigations that the Distributed \NVB scales well 
and depending on the macro mesh a fixed or very limited number global communications is needed. 


\begin{figure}[!ht]
  \includegraphics[width=0.18\textwidth]{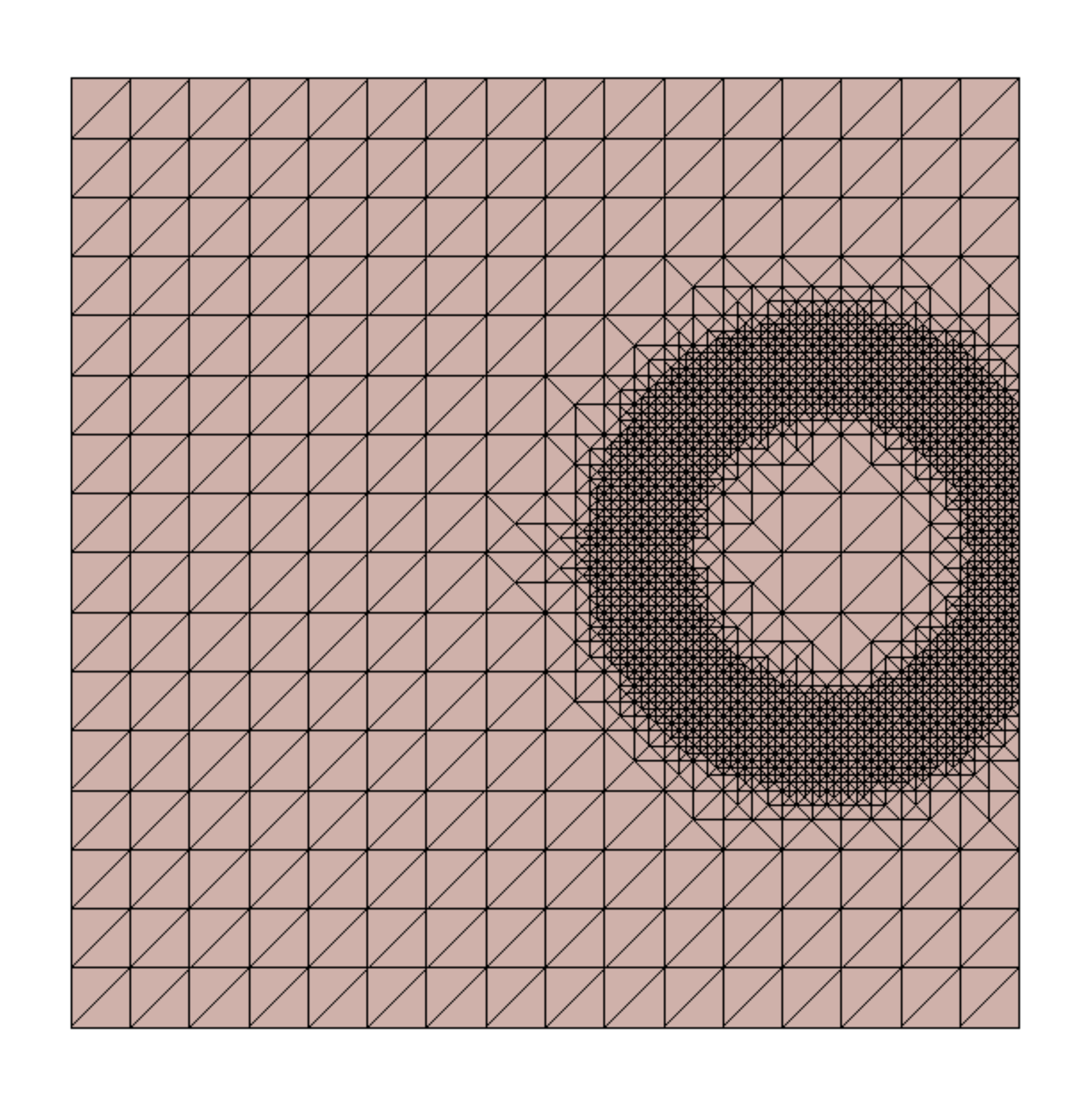}
  \includegraphics[width=0.18\textwidth]{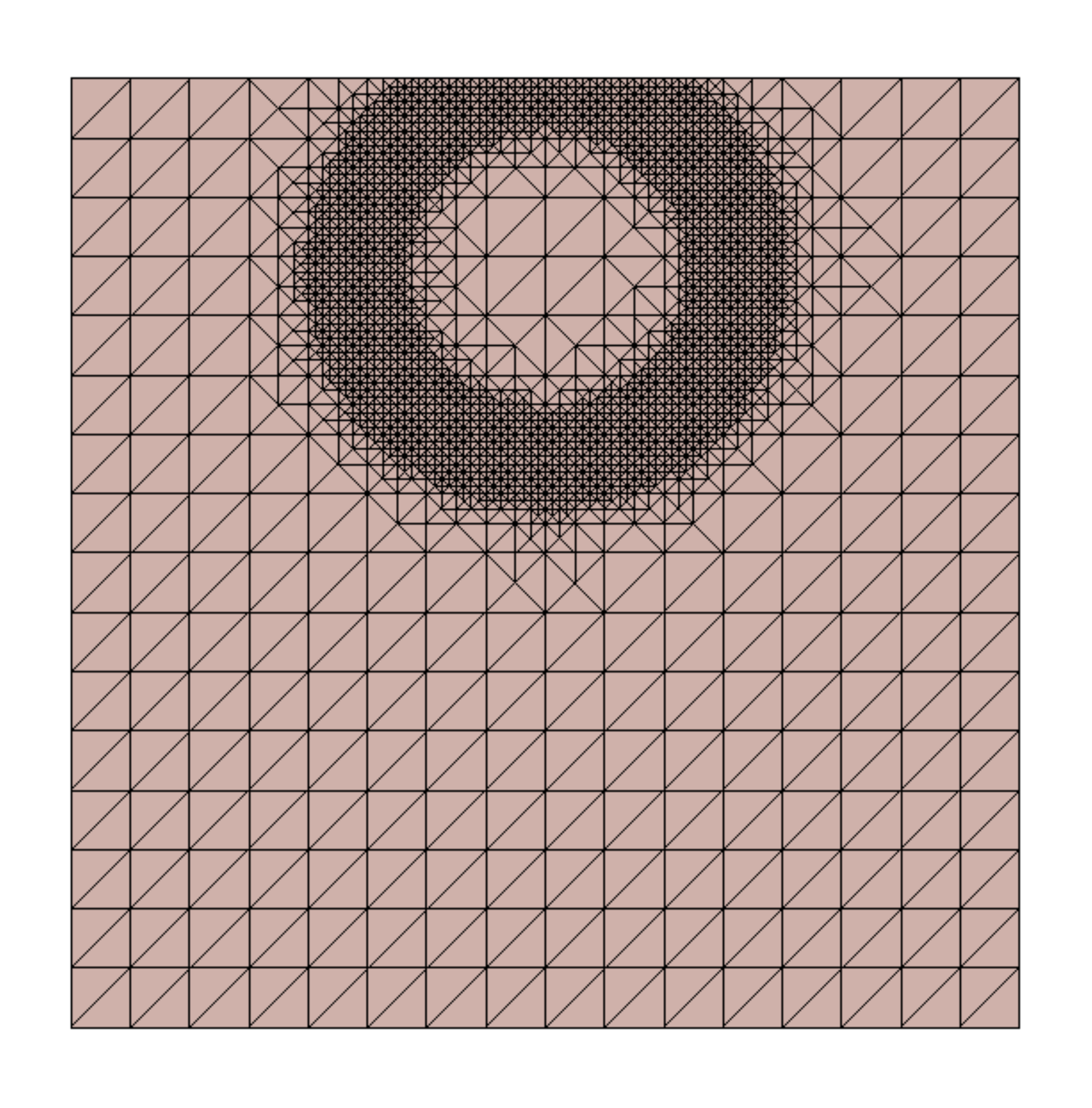}
  \includegraphics[width=0.18\textwidth]{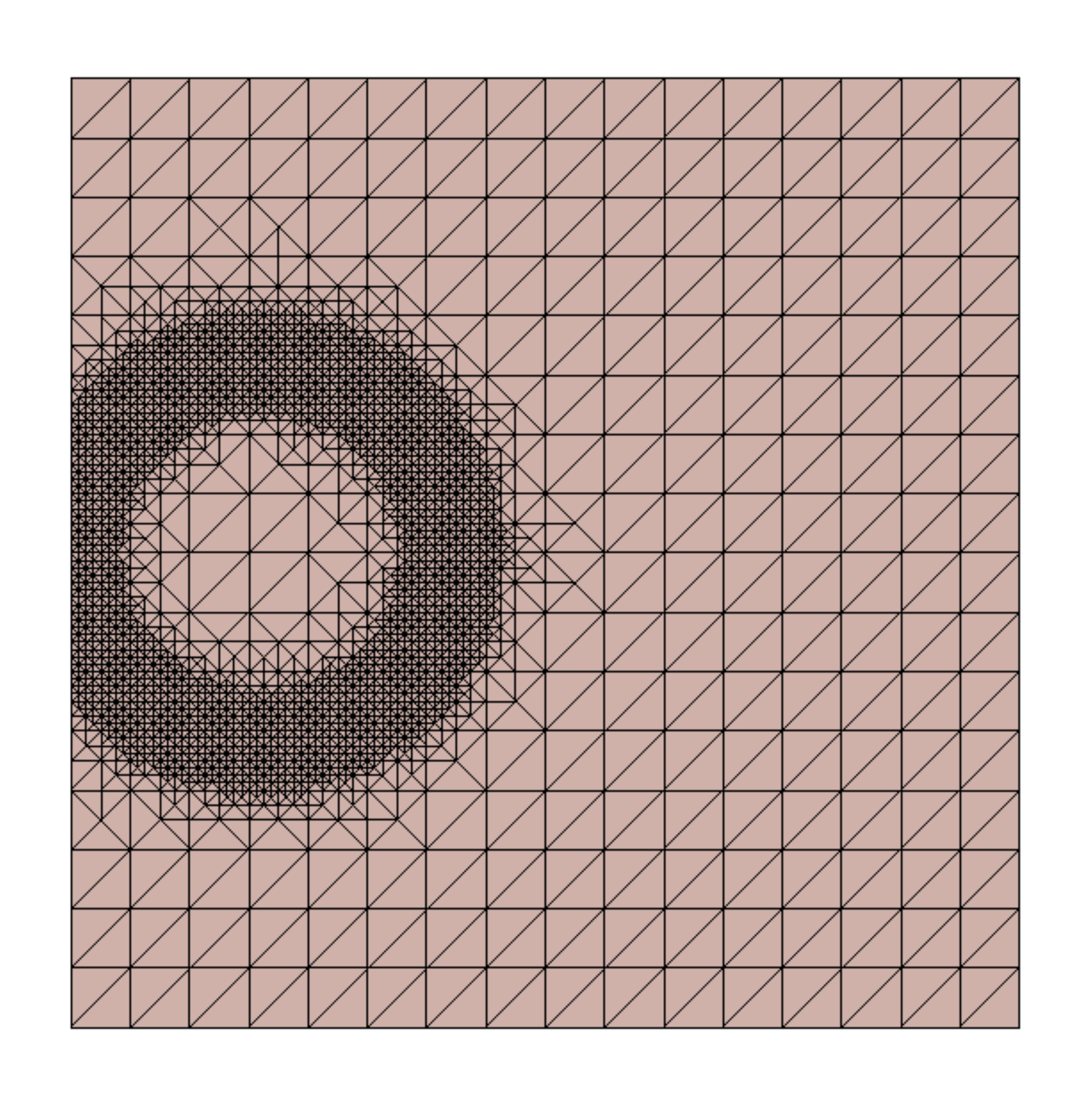}
  \includegraphics[width=0.18\textwidth]{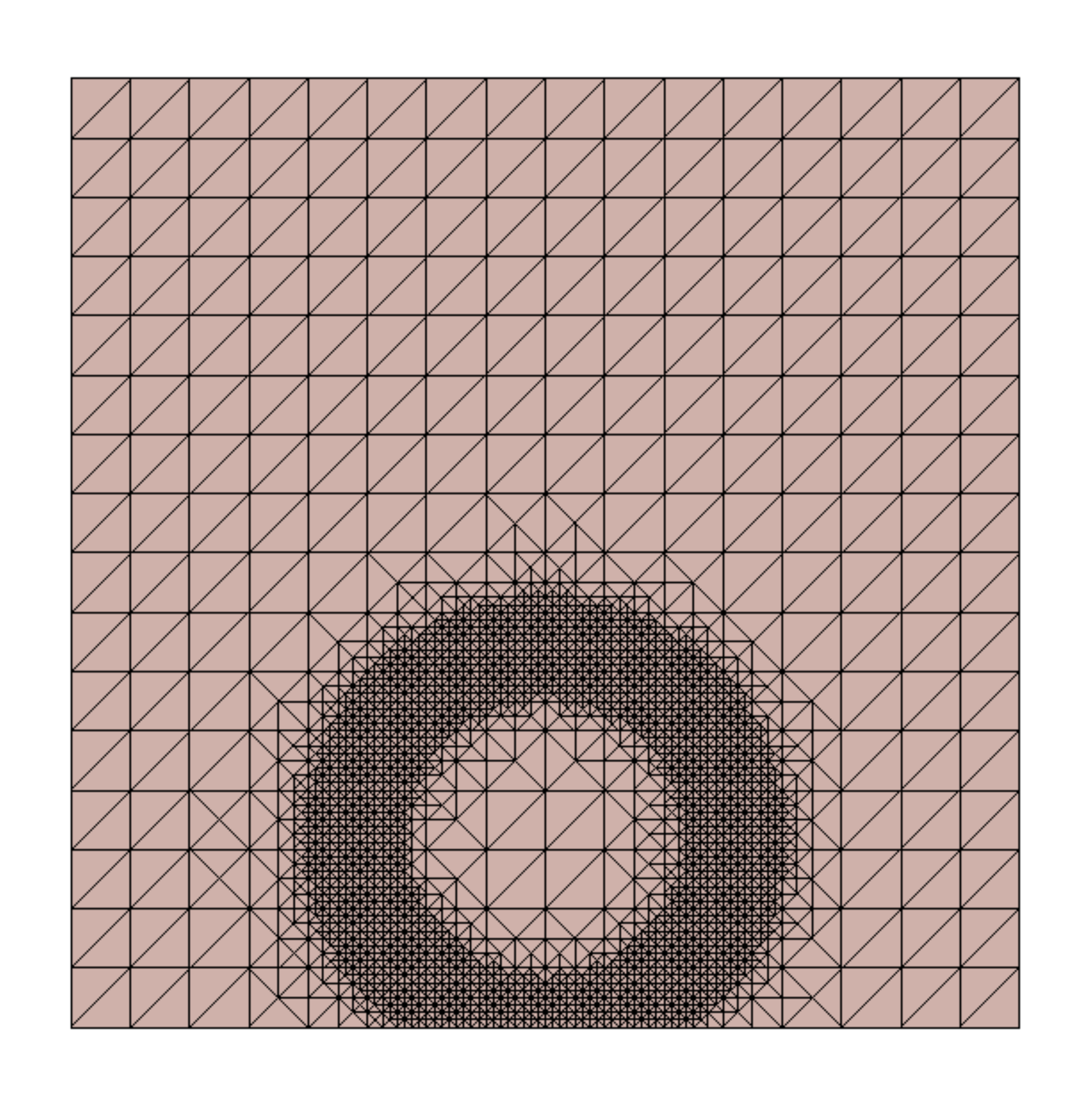}
  \includegraphics[width=0.18\textwidth]{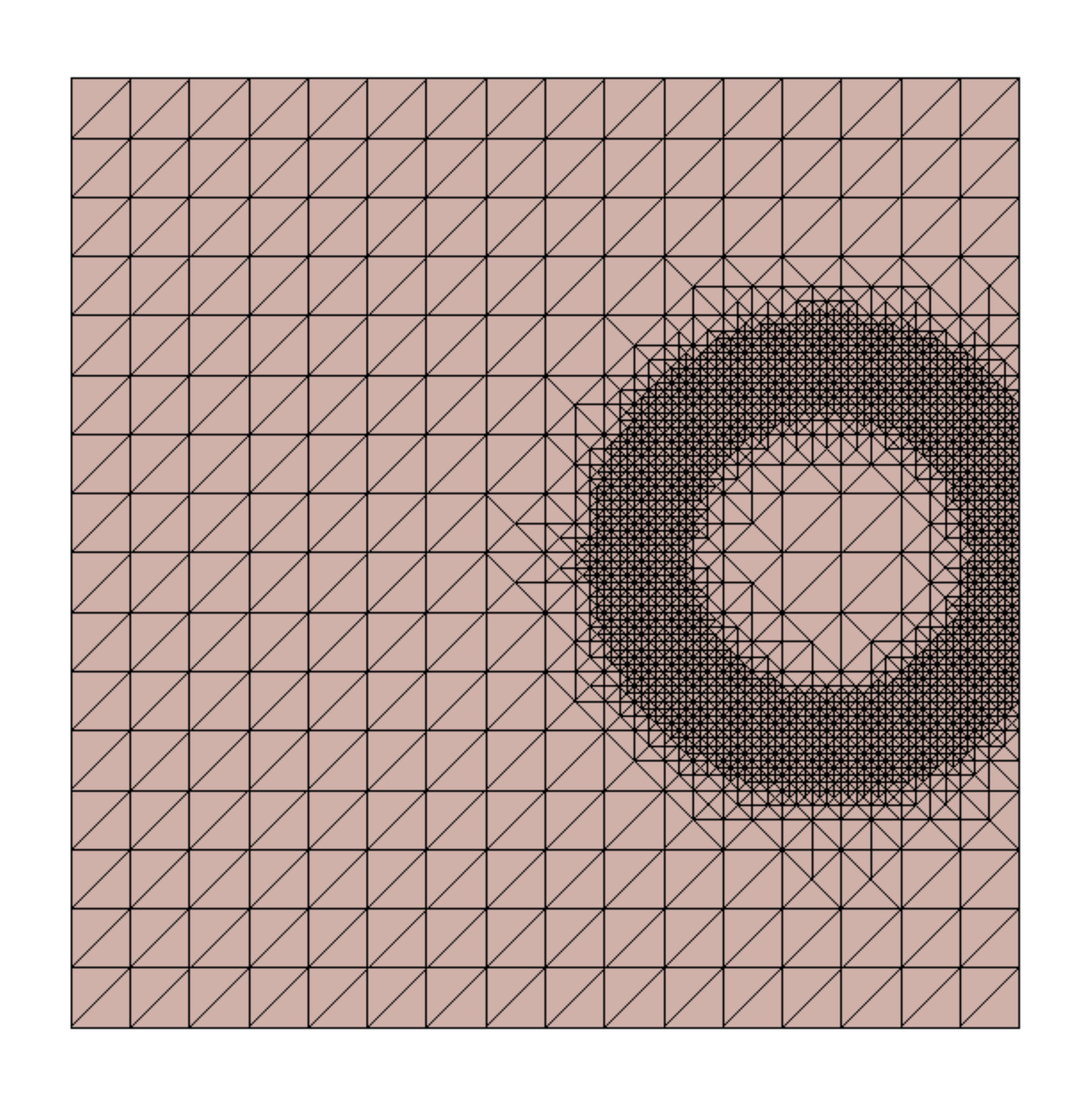}
  \caption{Refinement and coarsening of a doughnut like area 
    that rotates around the center. From left to
  right the simulation time is increased from $t=0$ to $t=1$ by $\Delta t = 0.25$.}
  \label{fig:2dball}
\end{figure}

\begin{figure}[!ht]
  \includegraphics[width=0.48\textwidth]{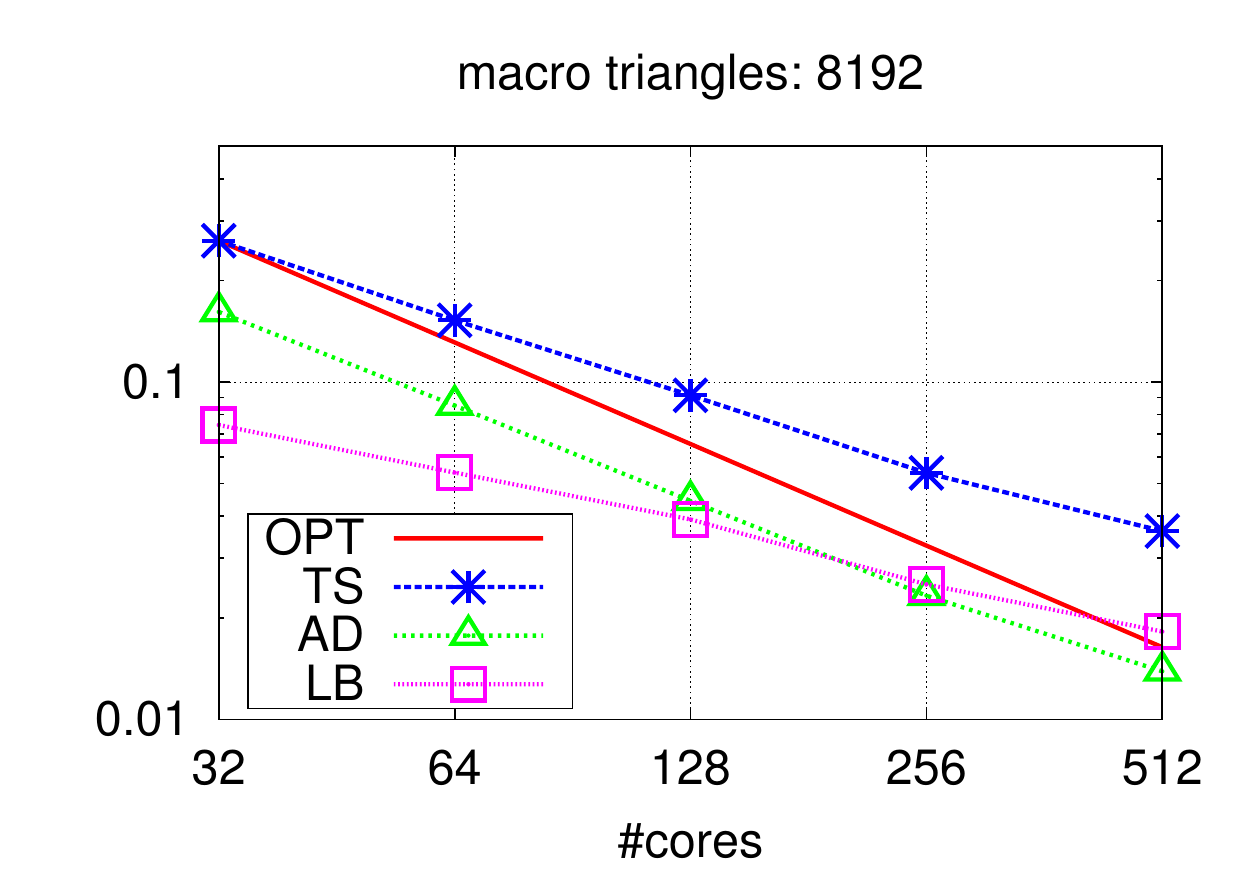}
  \includegraphics[width=0.48\textwidth]{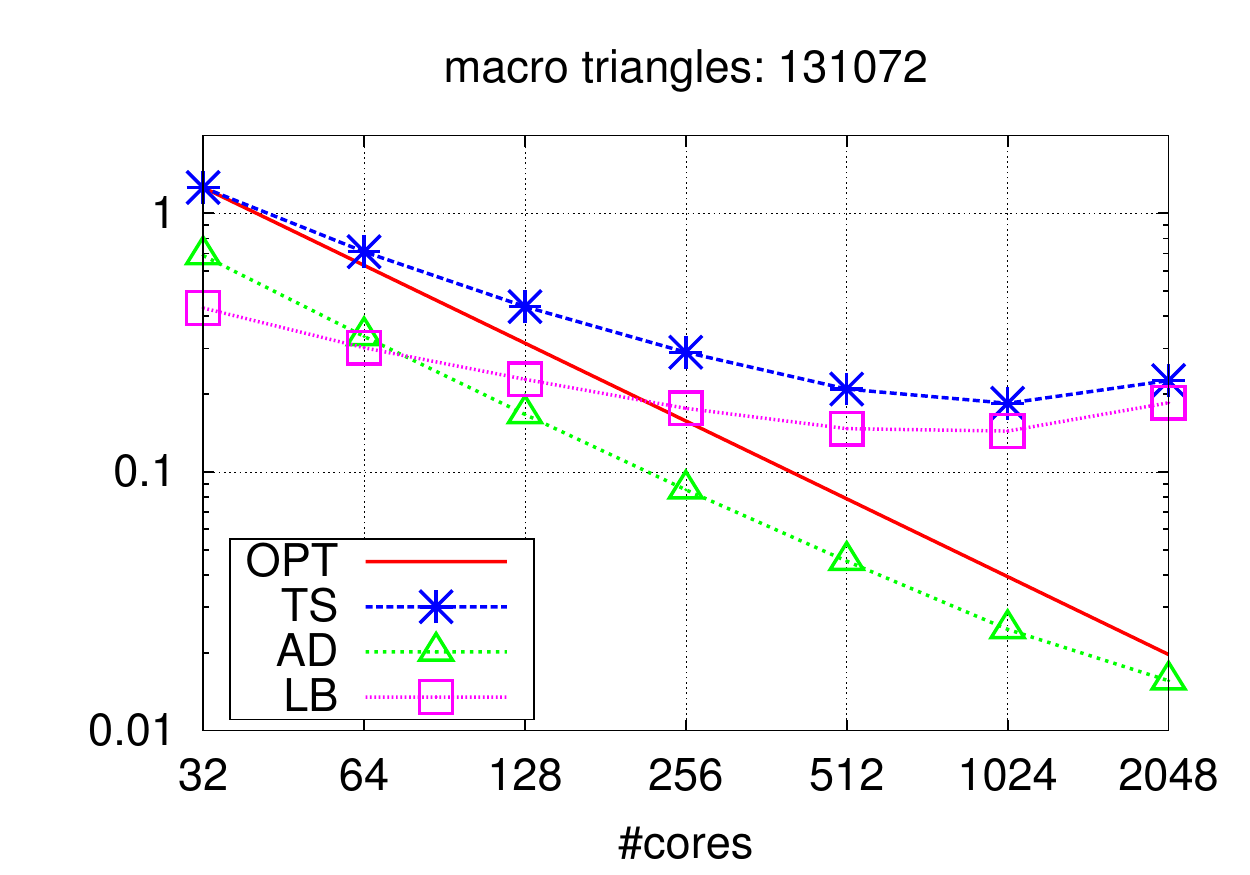}
  \caption{2d results for the doughnut refinement test on a coarse triangular macro mesh (left)
  and a finer triangular macro mesh (right). 
  For different number of cores the graph shows average run time per timestep in seconds 
  for the different parts of the algorithm: 
  a full time step (TS), the adaptation loop (AD),
  the load balancing (LB), and the expected optimal scaling (OPT).
  The scaling study has been performed on Yellowstone \cite{Yellowstone}.
  The test case is part of the \dune[ALUGrid] code and described in \cite{alugrid:16}.}
  \label{fig:2dscaling}
\end{figure}  

\begin{figure}[!ht]
  \includegraphics[width=0.48\textwidth]{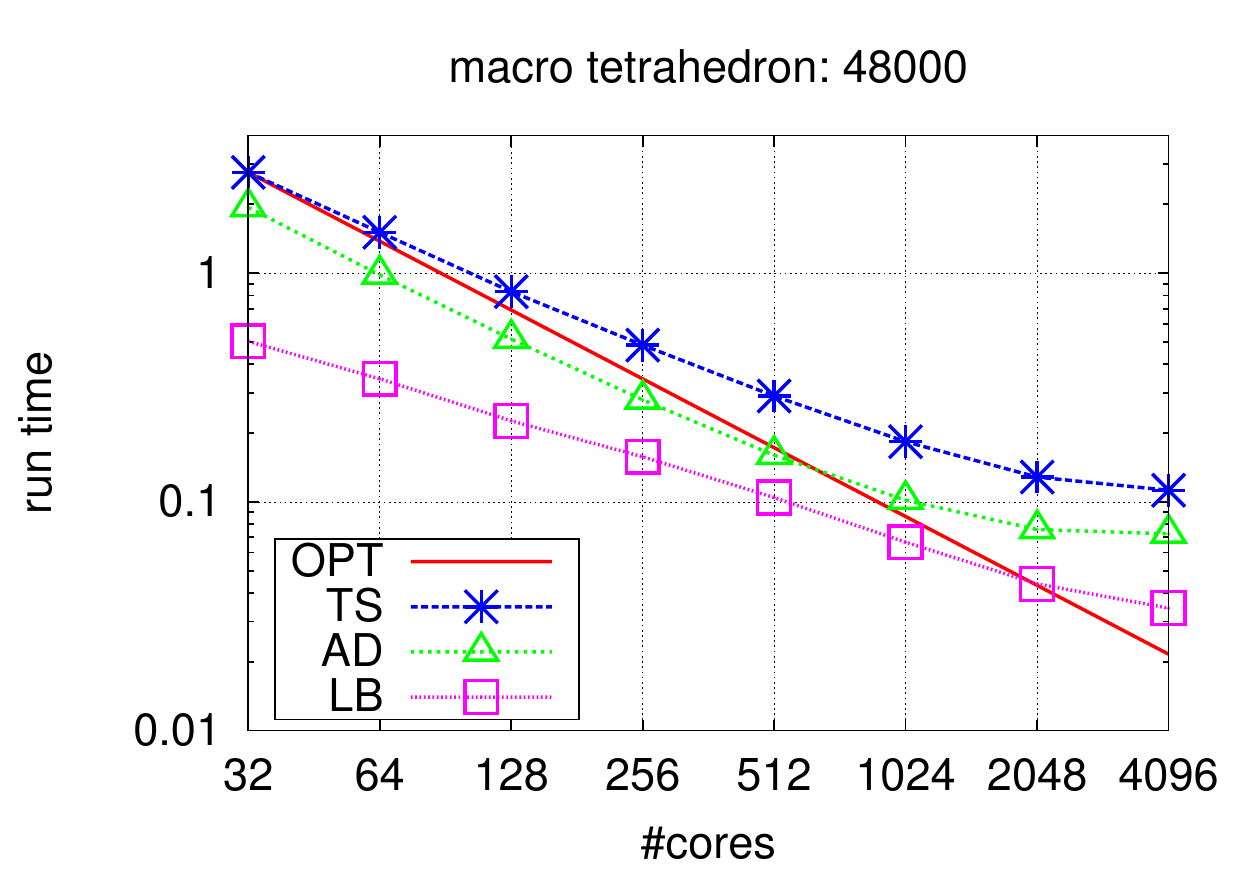}
  \includegraphics[width=0.48\textwidth]{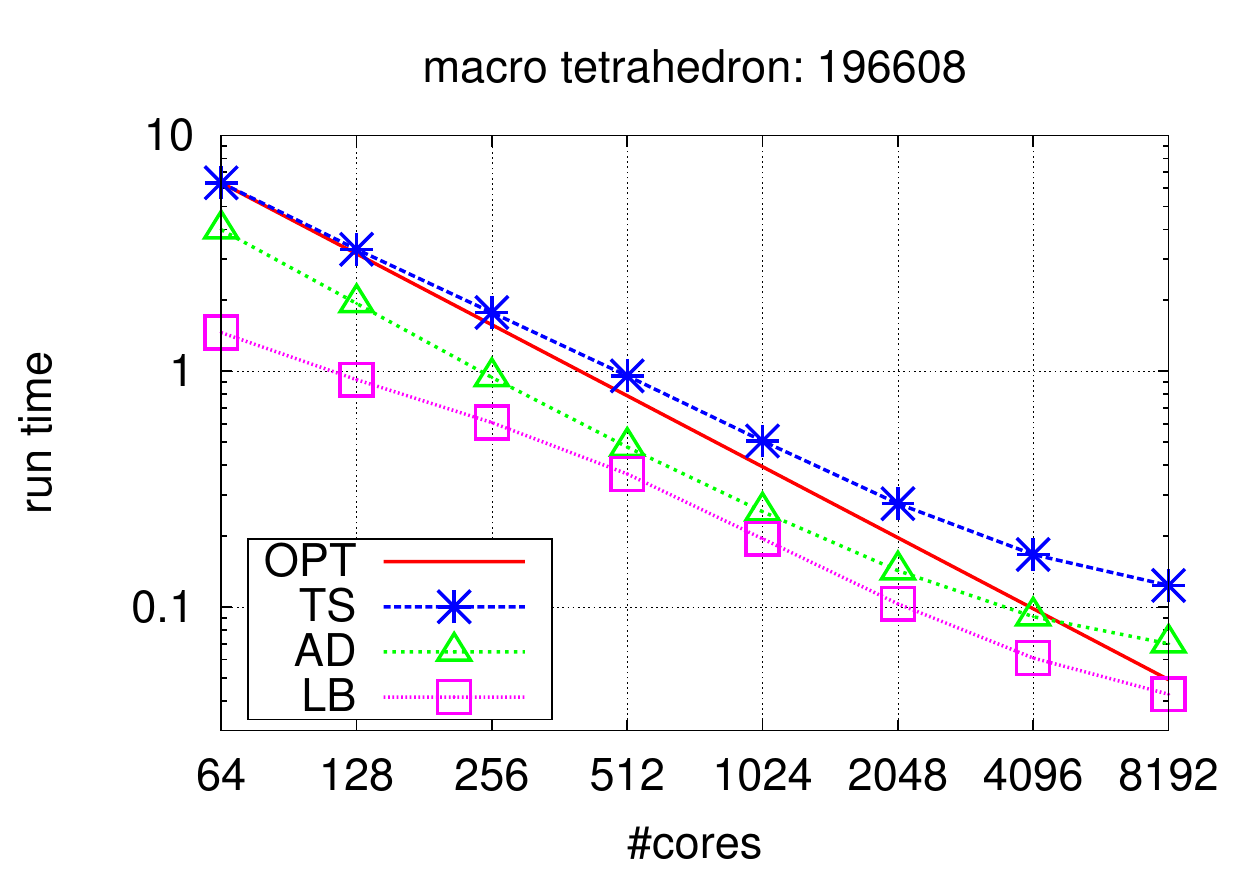}
  \caption{3d results for the doughnut refinement test on a coarse tetrahedral macro mesh (left) 
  and a finer tetrahedral macro mesh (right).
  For different number of cores the graph shows average run time per timestep in seconds 
  for the different parts of the algorithm: 
  a full time step (TS), the adaptation loop (AD),
  the load balancing (LB), and the expected optimal scaling (OPT).
  The scaling study has been performed on Yellowstone \cite{Yellowstone}.
  The test case is part of the \dune[ALUGrid] code and described in \cite{alugrid:16}.}
  \label{fig:3dscaling}
\end{figure}

\section{Summary}
We have shown (and proven in 2d), 
that the number of iterations in Algorithm \ref{algo:DNVB} to reach a 
conforming situation is bounded. In particular for grid implementations that do 
not partition the mesh on the leaf level, but on a certain fixed level, 
the bound is constant and independent of the current refinement situation. 
As the purpose of conforming grids is usually solving elliptic equations, 
the total run time is usually dominated by solving the equation. 
The performed \textit{worst-case} experiments are reflected by the theory. 
These experiments also prove that the compatibility condition in 
Assumption~\ref{A:initial_grid} is essential 
and cannot be neglected. 
In addition, we presented a state of the art implementation 
of the Distributed \NVB including asynchronous communication 
which is needed to achieve excellent scaling on a petascale super computer.

As a next step we will improve the flexibility of the load balancing algorithm which
currently only allows to partition the macro mesh. For certain problems where singularities
might arise, this might not be sufficient and yield poor scalability. 

\section*{Acknowledgements}

Martin Alk\"amper acknowledges the Cluster of Excellence in Simulation Technology (SimTech) at 
the University of Stuttgart for financial support.

Robert Kl\"ofkorn acknowledges NCAR/CISL's Research and Supercomputing Visitor Program (RSVP) and   
the Research Council of Norway and the industry partners --
ConocoPhillips Skandinavia AS, BP Norge AS, Det Norske Oljeselskap AS, Eni Norge AS, Maersk Oil Norway AS,
DONG Energy A/S, Denmark, Statoil Petroleum AS,
ENGIE E\&P NORGE AS, Lundin Norway AS, Halliburton AS,
Schlumberger Norge AS, Wintershall Norge AS -- of The National IOR Centre of Norway for financial support.

\end{document}